\documentclass[11pt]{article}
\usepackage{fullpage}
\pdfoutput=1	% needed for the arXiv to properly compile with pdflatex

\usepackage[tracking=smallcaps]{microtype}	% for an arXiv submission, set \pdfoutput=1 near the top so it uses pdflatex

\usepackage[utf8]{inputenc} % set input encoding (not needed with XeLaTeX)
\usepackage{amsmath,amsthm,amsfonts}
\usepackage{color}
\usepackage[usenames,dvipsnames] {xcolor}
\usepackage{url}
\definecolor{DarkGray}{rgb}{0.3,0.3,0.3}
\usepackage[colorlinks=true,breaklinks, linkcolor=DarkGray, citecolor=DarkGray, urlcolor=DarkGray]{hyperref}
\usepackage{tabularx}

\newcommand{\abs}[1]{\left| #1 \right|}

\def\Tr{\text{Tr}}

\newcommand{\bra}[1]{\left\langle{#1}\right\vert}
\newcommand{\ket}[1]{\left\vert{#1}\right\rangle}

\newtheorem{theorem}{Theorem}[section]
\newtheorem{lemma}[theorem]{Lemma}
\newtheorem{corollary}[theorem]{Corollary}

\newcommand{\eqnref}[1]{\hyperref[#1]{{(\ref*{#1})}}}
\newcommand{\thmref}[1]{\hyperref[#1]{{Theorem~\ref*{#1}}}}
\newcommand{\lemref}[1]{\hyperref[#1]{{Lemma~\ref*{#1}}}}
\newcommand{\corref}[1]{\hyperref[#1]{{Corollary~\ref*{#1}}}}
\newcommand{\defref}[1]{\hyperref[#1]{{Definition~\ref*{#1}}}}
\newcommand{\secref}[1]{\hyperref[#1]{{Section~\ref*{#1}}}}
\newcommand{\figref}[1]{\hyperref[#1]{{Fig.~\ref*{#1}}}}
\newcommand{\tabref}[1]{\hyperref[#1]{{Table~\ref*{#1}}}}
\newcommand{\remref}[1]{\hyperref[#1]{{Remark~\ref*{#1}}}}
\newcommand{\appref}[1]{\hyperref[#1]{{Appendix~\ref*{#1}}}}
\newcommand{\claimref}[1]{\hyperref[#1]{{Claim~\ref*{#1}}}}
\newcommand{\propref}[1]{\hyperref[#1]{{Proposition~\ref*{#1}}}}
\newcommand{\exampleref}[1]{\hyperref[#1]{{Example~\ref*{#1}}}}
\newcommand{\conjref}[1]{\hyperref[#1]{{Conjecture~\ref*{#1}}}}

\usepackage{graphicx}
\usepackage{caption}
\usepackage{subcaption}

\usepackage{bbm}
\newcommand{\Id}{I}
\newcommand{\KMMring}{\mathbb{Z}[i,\frac{1}{\sqrt 2}]}
\newcommand{\Aring}{\mathbb{Z}[i,\sqrt 2]}
\newcommand{\CNOT}{\text{CNOT}}
\newcommand{\Clifford}{\text{Clifford}}

\begin{document}
\title{Repeat-Until-Success: Non-deterministic decomposition of single-qubit unitaries}%
\author{
Adam Paetznick\\
David R. Cheriton School of Computer Science and\\
Institute for Quantum Computing,\\
University of Waterloo
\and Krysta M. Svore\\
Quantum Architectures and Computation Group,\\
Microsoft Research
}%
%\address{}%
%\thanks{}%
%\date{}
\maketitle
% ----------------------------------------------------------------
\begin{abstract}
We present a decomposition technique that uses non-deterministic circuits to approximate an arbitrary single-qubit unitary to within distance $\epsilon$ and requires significantly fewer non-Clifford gates than existing techniques.
We develop ``Repeat-Until-Success" (RUS) circuits and characterize unitaries that can be exactly represented as an RUS circuit.
Our RUS circuits operate by conditioning on a given measurement outcome and using only a small number of non-Clifford gates and ancilla qubits.
We construct an algorithm based on RUS circuits that approximates an arbitrary single-qubit $Z$-axis rotation to within distance $\epsilon$, where the number of $T$ gates scales as $1.26\log_2(1/\epsilon) - 3.53$, an improvement of roughly three-fold over state-of-the-art techniques.
We then extend our algorithm and show that a scaling of $2.4\log_2(1/\epsilon) - 3.28$ can be achieved for arbitrary unitaries and a small range of $\epsilon$, which is roughly twice as good as optimal deterministic decomposition methods.
\end{abstract}
% ----------------------------------------------------------------

% Open questions:
% \begin{itemize}
%   \item Can all multi-ancilla circuits be reduced to a single ancilla (and vice-versa)?  If so, what are the transformation rules?
%   \item Given a unitary that meets the basic conditions (i.e., scalar times elements in the ring) what are the corresponding repeat-until-success statistics, such as $T$ count, success probability, number of ancillas? Can we find an algorithm for synthesis?
% \end{itemize}

\section{Introduction
\label{sec:introduction}
}

As quantum devices continue to mature, there is an emerging need for algorithms that can efficiently and accurately map a high-level quantum algorithm into a low-level fault-tolerant circuit representation.
The mapping of a quantum algorithm into its equivalent fault-tolerant circuit representation requires first the choice of a universal basis or gate set, and second a decomposition algorithm that can translate a quantum circuit into a sequence of gates drawn from that basis.
The choice of basis is predominantly dictated by the existence of resource-efficient, fault-tolerant quantum error correction protocols for each gate; a common set is CNOT plus the universal single-qubit basis $\{H,T\}$, where $H= \frac{1}{\sqrt{2}} \left(\begin{smallmatrix} 1&1\\1&-1 \end{smallmatrix}\right)$ and $T = \left(\begin{smallmatrix} 1&0\\0&e^{i\pi/4} \end{smallmatrix}\right)$.
For many quantum error-correcting codes, a fault-tolerant $H$ requires transversal application of the gate, and a fault-tolerant $T$ requires magic state distillation.
The cost of a $\{H,T\}$ circuit is defined to be the number of $T$ gates, given that the resource cost of a fault-tolerant $T$ gate is up to an order of magnitude larger than the resource cost of a fault-tolerant $H$ gate~\cite{Raussendorf2007a,Fowler2013}.

The decomposition algorithm should minimize the desired cost function, such as the $T$ count of the $\epsilon$-approximate gate sequence.
The Solovay-Kitaev theorem~\cite{Kitaev1997a,Kita02}, guarantees that a single-qubit unitary operation can be efficiently approximated to within error $\epsilon$ by a sequence of $O(\log^c(1/\epsilon))$ gates from a discrete universal basis, where $c = 1$ is the theoretical lower bound~\cite{Knill1995}.
Fowler gave an exponential-time algorithm that achieves the lower bound, resulting in an approximating sequence containing $2.95\log_2(1/\epsilon) + 3.75$ $T$ gates, on average~\cite{Fowl04c}.
However, the exponential time complexity limits the achievable accuracy.
A database search algorithm based on canonical forms for $\{H,T\}$ circuits was given by Bocharov and Svore \cite{Bocharov2012} that also achieves the lower bound and enables search to slightly better accuracy.
Recently, efficient algorithms that achieve the lower bound have been developed.
Kliuchnikov, Maslov and Mosca (KMM) developed an algorithm which yields $3.21\log_2(1/\epsilon) - 6.93$ $T$ gates for the rotation $R_Z(1/10)$~\cite{Kliuchnikov2012b}.
Selinger's algorithm $\epsilon$-approximates a single-qubit $Z$-axis rotation, $R_Z(\theta) = \left(\begin{smallmatrix} 1&0\\0&e^{i\theta} \end{smallmatrix}\right)$, using  $4\log_2(1/\epsilon) + 11$ $T$ gates in the worst case~\cite{Selinger2012a}. Subsequent improvement by Ross and Selinger yields a scaling of $3\log(1/\epsilon) + O(\log\log(1/\epsilon))$ in typical cases~\cite{Ross2014}.

For a given single-qubit unitary $U$ and error $\epsilon$, the above algorithms output a fixed sequence of single-qubit gates from the set $\{H,T\}$, without the use of ancillary qubits or measurements.
In this paper, we present a circuit framework and algorithm to minimize the $T$ gates required to approximate a given single-qubit unitary.
We show that by incorporating ancilla qubits and measurements, the expected number of $T$ gates required to approximate a random $Z$-axis rotation can be significantly reduced to
\begin{equation}
\label{eq:axial-tcount-scaling}
\text{Exp}_Z[T] = 1.26\log_2(1/\epsilon) - 3.53
\enspace ,
\end{equation}
an improvement of roughly three-fold over~\cite{Selinger2012a} and more than two-fold over~\cite{Fowl04c}, \cite{Kliuchnikov2012b} and~\cite{Ross2014}.
For arbitrary single-qubit unitaries, our results indicate a significantly reduced $T$-count scaling of
\begin{equation}
\label{eq:nonaxial-tcount-scaling}
\text{Exp}_U[T] = 2.4\log_2(1/\epsilon) - 3.28
\enspace ,
\end{equation}
roughly $50$ percent better than using~\eqnref{eq:axial-tcount-scaling} for each $Z$ rotation (three are required in general) and up to four-fold better than traditional ancilla-free decomposition.% of three $Z$-axis rotations.

Our circuits are distinct from those output by Fowler, KMM and Selinger in that they are \textit{non-deterministic}.
Each circuit, when conditioned on a particular measurement outcome, exactly implements a desired unitary, and otherwise implements a unitary that can be reversed at little or no cost; it can then be repeated until the desired unitary is obtained.
We call our circuits ``Repeat-Until-Success" (RUS) circuits.
A significant advantage of RUS circuits is the extremely low resource cost, in non-Clifford gates and ancillary qubits.

Our paper is structured as follows.
We begin in \secref{sec:prev} by discussing existing single-qubit unitary decomposition techniques, and the presence of RUS circuits in previous work.
We then characterize unitaries that can be exactly implemented ($\epsilon=0$) as an RUS circuit in \secref{sec:rus}.
In \secref{sec:search}, we present an optimized direct search algorithm for synthesizing RUS circuits with extremely low $T$ count and in \secref{sec:results}, we construct a corresponding database of RUS circuits.
Leveraging our database, we develop a decomposition algorithm to approximate a given unitary using compositions of RUS circuits in \secref{sec:applications}.
We then present a variety of applications of RUS circuits, including a circuit for the $V_3$ gate that results in state-of-the-art single-qubit decomposition.
Finally, we discuss future directions and open problems in \secref{sec:concl}.

\section{Existing methods for single-qubit unitary decomposition}
\label{sec:prev}
In addition to the techniques discussed above~\cite{Dawson2005,Fowl04c,Kliuchnikov2012b,Selinger2012a,Ross2014}, a variety of other methods for single-qubit unitary decomposition have been developed.
So-called ``phase kickback'' involves preparing a special ancilla state based on the quantum Fourier transform and then using phase estimation~\cite{Kita02}.
Non-deterministic circuits called ``programmable ancilla rotations''(PAR) use a cascading set of prepared ancilla states along with gate teleportation~\cite{Jones2012}.
Similar use of non-deterministic circuits to produce a ``ladder'' of non-stabilizer states, and in turn to approximate an arbitrary single-qubit unitary, has also been proposed \cite{Duclos-Cianci2012}.
The number of $T$ gates required for these ancilla-based methods is larger than for ancilla-free methods, but the total resources are comparable in some architectures~\cite{Jones2013b}.
For this reason, we compare our results to the Fowler, KMM, Selinger, and Ross-Selinger methods.

Non-deterministic circuits have also been proposed for decomposition into alternate gate sets.
Bocharov, Gurevich and Svore (BGS) showed that arbitrary single-qubit unitaries can be approximated using the gate set $\{H,S=T^2,V_3\}$, where $V_3 = (I + 2iZ)/\sqrt{5}$, with a typical scaling of $3\log_5(1/\epsilon)$
in the number of $V_3$ gates~\cite{Bocharov2013}.  They suggest a fault-tolerant implementation of the $V_3$ gate (see ~\figref{fig:v3-toffoli}) using an RUS circuit which requires eight $T$ gates, four for each Toffoli (see~\cite{Jones2012d}).  Later, Jones improved this circuit, using only a single Toffoli gate~\cite{Jones2013b}.
Using our optimized direct search algorithm, we find an improved RUS circuit for $V_3$ that uses only four $T$ gates, as shown in~\figref{fig:v3-5.26}, and is exact ($\epsilon=0$).
By contrast, an approximation to within $\epsilon = 10^{-6}$ using the KMM algorithm requires $67$ $T$ gates.
Furthermore, when used to implement $V_3$, our circuit results in $\{H,S,V_3\}$-decomposition achieving substantially lower $T$ count (on average) than $\{H,T\}$-decomposition methods.

Repeat-until-success circuits have also been used by Wiebe and Kliuchnikov \cite{Wiebe2013}, who proposed a family of tree-like, hierarchical RUS circuits that yield $T$ counts superior to Selinger and KMM for small-angle $Z$-axis rotations.
In contrast, our results show that RUS circuits can be used for large- and small-angle $Z$-axis rotations, as well as rotations about an arbitrary axis.
We also provide a general characterization of RUS circuits, and a general framework for their construction.

A summary of the $T$ count costs of our method, labeled RUS, and the above algorithms is given in Tables~\ref{tab:non-axial-decomposition-methods} and \ref{tab:axial-decomposition-methods} for non-axial and axial rotations, respectively.

\begin{table}
\centering
\begin{tabular}{>{\raggedright\arraybackslash}m{2.7cm}|>{\raggedright\arraybackslash}m{4.2cm}|>{\raggedright\arraybackslash}m{3.5cm}|>{\raggedright\arraybackslash}m{4.2cm}}
\hline
\textbf{Method}        & \textbf{Description} & \textbf{$T$ count} & \textbf{Comments}\\
\hline
Solovay-Kitaev \cite{Dawson2005}& Converging $\epsilon$-net \qquad\qquad based on group commutators. & $O(\log^{3.97}1/\epsilon)$& Computationally efficient, but sub-optimal $T$ count.\\
\hline
Ladder states \cite{Duclos-Cianci2012}& Hierarchical distillation based $\ket H$ states.& $O(\log^{1.75}1/\epsilon)$& Some of the cost can be shifted ``offline".\\
\hline
Direct search \cite{Fowl04c,Bocharov2012} & Optimized exponential-time search.& $2.95\log_2(1/\epsilon)+3.75$&Optimal ancilla-free $T$ count.\\
\hline
BGS \cite{Bocharov2013} & Direct search decomposition with $V_3$. & $T_V(3 \log_5 1/\epsilon)$& $T_V$ is the $T$ count for choice of fault-tolerant implementation of $V_3$.\\
\hline
\textbf{RUS} (non-axial)&Database lookup.& $2.4\log_2(1/\epsilon) - 3.28$ &Limited approximation accuracy.\\
\hline
\end{tabular}
\caption[Decomposition methods for arbitrary single-qubit unitaries.]{\label{tab:non-axial-decomposition-methods}
Decomposition methods for arbitrary single-qubit unitaries using the gate set $\{H,S,T\}$.
}
\vspace{1cm}
\begin{tabular}{>{\raggedright\arraybackslash}m{2.7cm}|>{\raggedright\arraybackslash}m{4.2cm}|>{\raggedright\arraybackslash}m{3.5cm}|>{\raggedright\arraybackslash}m{4.2cm}}
\hline
\textbf{Method}        & \textbf{Description} & \textbf{$T$ count} & \textbf{Comments}\\
\hline
Phase kickback \cite{Kita02}&Uses Fourier states and phase estimation.&$O(\log 1/\epsilon)$ ~~\qquad (implementation dependent)&$O(\log 1/\epsilon)$ ancillas. Optimizations make it cost competitive with Selinger and KMM.\\
\hline
PAR \cite{Jones2012}& Cascading gate teleportation.& $O(\log 1/\epsilon)$&Constant depth (on average), higher $T$ count than phase kickback.\\
\hline
Selinger \cite{Selinger2012a}&Round-off followed by exact decomposition.& $4\log(1/\epsilon)+11$&$T$ count is optimal for worst-case rotations.\\
\hline
Ross-Selinger \cite{Ross2014}&Round-off followed by exact decomposition.& $3\log(1/\epsilon) + O(\log\log 1/\epsilon)$& $T$ count is near-optimal for typical rotations.\\
\hline
KMM \cite{Kliuchnikov2012b}& Round-off followed by exact decomposition.& $3.21\log_2(1/\epsilon)-6.93$& $T$ count based on scaling for $R_Z(1/10)$.\\
\hline
Floating-point \cite{Wiebe2013}& A family of tree-like RUS circuits&$1.14\log_2(10^\gamma) + 8\log_2(10^{-\gamma}/\epsilon)$& For small angle $\theta = a\times 10^{-\gamma}$, $T$ count is roughly $1.14\log_2(1/\theta)$.\\
\hline
\textbf{RUS} (axial)& Database lookup.& $1.26\log_2(1/\epsilon)-3.53$&Approximation to within $\epsilon=10^{-6}$.\\
\hline
\end{tabular}
\caption[Decomposition methods for $Z$-axis rotations.]{\label{tab:axial-decomposition-methods}
Decomposition methods for $Z$-axis rotations using the gate set $\{H,S,T\}$.  Approximation of an arbitrary single-qubit unitary is possible by using the relation $U = R_Z(\theta_1)H R_Z(\theta_2) H R_Z(\theta_3)$.
}
\end{table}

RUS circuits have been considered in other contexts, as well.
The term was first used by~\cite{Lim2004} to describe the implementation of a CZ gate by repeated operations in linear optics.
More recently, \cite{Shah2013} adapted deterministic ancilla-driven methods~\cite{Anders2009,Kashefi2009} to allow for non-determinism.
Our use of repetition is similar to~\cite{Lim2004} and~\cite{Shah2013}, but we generate a family of circuits each of which are intended for use in conjunction with a fault-tolerant gate set, rather than at the physical level.

\section{Repeat-Until-Success circuits
\label{sec:rus}
}

To describe RUS circuits, we begin with an example.
Consider the circuit shown in~\figref{fig:v3-toffoli}, which performs the single-qubit unitary $V_3 = (I + 2iZ)/\sqrt{5}$.
This circuit involves two measurements in the Pauli $X$-basis.  If both measurement outcomes are zero, then the output is equivalent to $V_3\ket\psi$.  If any other outcome occurs, then the output is $I\ket\psi = \ket\psi$.  Thus, the circuit may be repeated until obtaining the all zeros outcome, and the number of repetitions will vary according to a geometric probability distribution. (In this case the probability of getting both zeros is $5/8$.)  Upon measuring all zeros, the unitary $V_3$ is implemented $\emph{exactly}$, even though the overall circuit is non-deterministic.

\begin{figure}[tb]
\centering
\begin{subfigure}[b]{.3\textwidth}
\centering
\includegraphics[scale=1]{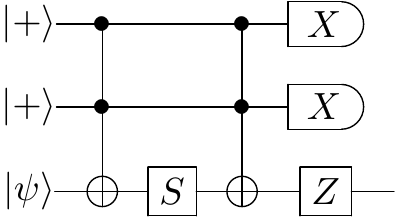}
\caption{\label{fig:v3-toffoli}
$\text{Exp}[$T$]=12.8$}
\end{subfigure}
\hfill
\begin{subfigure}[b]{.3\textwidth}
\centering
\includegraphics[scale=1]{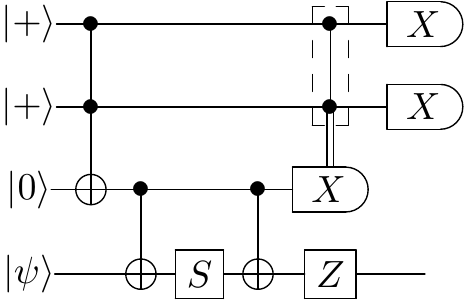}
\caption{\label{fig:v3-jones}
$\text{Exp}[$T$]=6.4$}
\end{subfigure}
\hfill
\begin{subfigure}[b]{.3\textwidth}
\centering
\includegraphics[scale=1]{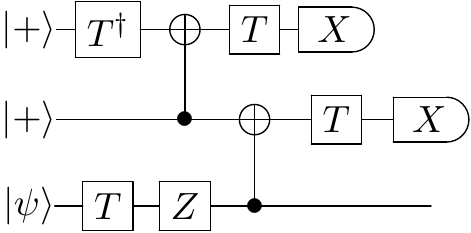}
\caption{\label{fig:v3-5.26}
$\text{Exp}[$T$] < 5.26$}
\end{subfigure}
\caption{
\label{fig:v3-circuits}
Repeat-Until-Success circuits for $V_3 = (I + 2iZ)/\sqrt{5}$.
Each of the circuits above implements $V_3$ conditioned on an $X$-basis measurement outcome of zero on each of the top two ancilla qubits. If any other measurement outcome occurs, then each circuit implements the identity. The probability of measuring $00$ is $5/8$ for each circuit.  Repeating the circuit until success yields an expectation value for the number of $T$ gates, as indicated.
(a) A slight modification of the circuit presented in~\cite{Nielsen2000} pp. $198$. Each Toffoli gate can be implemented with four $T$ gates (see~\cite{Jones2012d}).
(b) A circuit proposed by Jones that requires just a single Toffoli gate~\cite{Jones2013b}.
(c) An alternative circuit found by direct search. Measurement of the first qubit can be performed before interaction with the data qubit.  Thus the top-left part of the circuit can be repeated until measuring zero.  The probability of measuring zero on the first qubit is $3/4$.  The probability of measuring zero on the second qubit, conditioned on zero outcome of the first qubit, is $5/6$.  The $T$ gate applied directly to $\ket\psi$ can be freely commuted through the CNOT.  In the case that an even number of attempts are required, the $T$ gates can be combined into the Clifford gate $T^2 = S$.
}
\end{figure}

\begin{figure}[tb]
\centering
\includegraphics{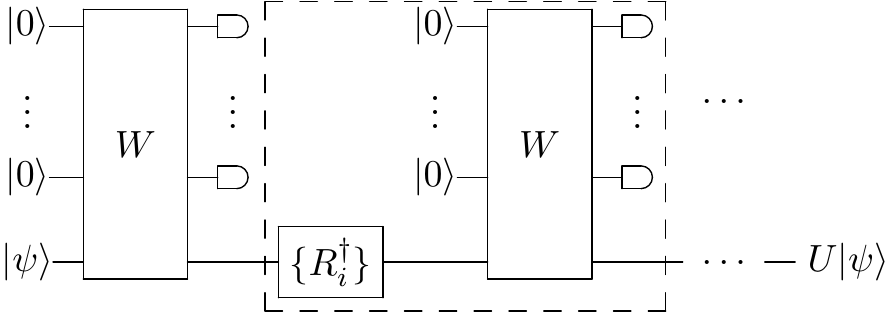}
\caption{\label{fig:rus-circuit}
A Repeat-Until-Success circuit that implements the unitary $U$.  Ancilla qubits are prepared in $\ket 0$, then the unitary $W$ is performed on both the ancillas and $\ket\psi$.  Upon measuring the ancillas, a unitary operation is effected on $\ket\psi$ which is either $U$ or one of $\{R_i\}$, depending on the measurement outcome.  If the measurement outcome indicates $R_i$, then the recovery operation $R_i^\dagger$ is performed, and the process can be repeated.
}
\end{figure}

We define a Repeat-Until-Success (RUS) circuit over a gate set $G$ to be of the following general structure:
\begin{enumerate}
\itemsep 1pt
 \parskip 0pt
\item Prepare $m$ ancilla qubits in the state $\ket{0^m}$.
\item Given an input state $\ket\psi$ on $n$ qubits, apply a unitary $W$ to all of the $n + m$ qubits using gates from $G$.
\item Measure each ancilla qubit in the computational basis.  The output is given by $\Phi_i\ket\psi$, where $\Phi_i$ is a quantum channel on $n$ qubits that depends on the measurement outcome $i \in \{0,1\}^m$.
\item If the measurement outcome indicates ``failure", apply a recovery operation and repeat.
\end{enumerate}

The measurement outcomes are partitioned into two sets: ``success'' and ``failure''.
Success corresponds to some set of desired operations $\{\Phi_i : i \in \text{success}\}$;
failure corresponds to some set of undesired operations $\{\Phi_i : i \in \text{failure}\}$.
In the case of success, no further action is required.
In the case of failure $i$, a recovery operation $\Phi_i^{-1}$ is applied, and the circuit is repeated.
For practical purposes, the recovery operations should be implementable for relatively low cost compared to $W$.

We restrict to the case in which $\ket\psi$ is a single qubit and the $\{\Phi_i\}$ are unitary.
We also limit to a single ``success'' output $U\ket\psi$, for some unitary $U$, though $U$ may correspond to multiple measurement outcomes.
The operation $W$ is then given by a $2^{m+1}\times 2^{m+1}$ unitary matrix of the form
\begin{equation}
\label{eq:nondeterministic-unitary-form}
W =
\frac{1}{\sqrt{\sum_i \abs{\alpha_i}^2}}
\begin{pmatrix}
\alpha_0 U & \ldots \\
\alpha_1 R_1 & \ddots \\
\vdots & \\
\alpha_l R_{l} &\\
\end{pmatrix}
\enspace ,
\end{equation}
where $U,R_1,\ldots,R_l$ are $2\times 2$ unitary matrices, and $\alpha_0,\ldots,\alpha_l \in \mathbb{C}$ are scalars.
Since the ancillas are prepared in $\ket{0^m}$, only the first two columns of $W$ are of consequence.
Contents of the remaining columns are essentially unrestricted, except that $W$ must be unitary.
Each of the $l + 1 = 2^{m}$ measurement outcomes corresponds to application of a unitary from $U \cup \{R_i\}$ on the input qubit $\ket\psi$.  Without loss of generality, we select the all zeros outcome to correspond with application of $U$, since outcomes can be freely permuted.
The entire protocol is illustrated in~\figref{fig:rus-circuit}.

% For simplicity, we assume that $U \neq R_i ~\forall~ 1\leq i \leq l$.
% The case in which $U$ appears multiple times can be easily accommodated.
% In order for the circuit to be useful, the remaining matrices $R_1,\ldots,R_l$ should be invertible at a low cost.

To ensure compatibility with existing fault-tolerance schemes, we require that $W$ can be synthesized using the gate set $G=\{\Clifford,T\}$, where $\Clifford$ denotes the Clifford group generated by $\{H, S, \CNOT\}$; note that our framework and algorithms can be extended to other gates sets with little difficulty.
A unitary matrix is exactly implementable by $\{\Clifford,T\}$ if and only if its entries are contained in the ring extension $\KMMring$~\cite{Giles2012}.  Thus, we require that $\alpha_0 U,\alpha_1 R_1, \ldots,\alpha_l R_l \in \KMMring$.
Furthermore, the normalization $1/\sqrt{\sum_i \abs{\alpha_i}^2}$ must also be in the ring.
The unitarity condition on $W$ then requires that
\begin{equation}
\label{eq:HT-unitarity-condition}
\sum_i \abs{\alpha_i}^2 = 2^{k}
\end{equation}
for some integer $k$.

If all of the recovery operations $R_1,\ldots,R_l$ are exactly implementable by $\{\Clifford,T\}$, then we may assume that $\alpha_1,\ldots,\alpha_l\in \KMMring$.
If $\alpha_0$ is an integer, then Lagrange's four-square theorem implies that~\eqnref{eq:HT-unitarity-condition} can be satisfied using at most $m=2$ ancilla qubits.

\subsection{Characterization
\label{rus.characterization}
}
Consider a $2\times 2$ unitary matrix $U$ such that
\begin{equation}
\label{eq:rus-unitary}
U = \begin{pmatrix} u_{00} & u_{01} \\ u_{10} & u_{11}\end{pmatrix}
  = \frac{1}{\sqrt{2^k}\alpha}\begin{pmatrix} \beta_{00} & \beta_{01} \\ \beta_{10} & \beta_{11}\end{pmatrix}
  \enspace ,
\end{equation}
for $\alpha \in \mathbb{R}$, $\beta_{00},\ldots,\beta_{11} \in \Aring$ and integer $k \geq 0$.
We are concerned with exactly implementing $U$ only up to a global unit phase $e^{i \phi}$ for some $\phi \in [0,2\pi)$. Accordingly, we may assume without loss of generality that $\alpha$ is real and non-negative since for any $\beta \in \mathbb{C}$, $\frac{\beta \beta^*}{\abs{\beta}} \geq 0$.
The restriction to $\Aring$ rather than $\KMMring$ is also without loss of generality, since $k$ can be chosen to eliminate any denominators.
Then choosing $\alpha_0=\sqrt{2^k}\alpha$ we have
\begin{equation}
\label{eq:rus-scalar-condition}
\alpha_0 = \sqrt{\abs{\beta_{00}}^2 + \abs{\beta_{10}}^2} = \sqrt{x + y\sqrt 2}
\enspace ,
\end{equation}
where $x = a_{00}^2 + c_{00}^2 + a_{10}^2 + c_{10}^2 + 2(b_{00}^2 + d_{00}^2 + b_{10}^2 + d_{10}^2)$, $y = a_{00}b_{00} + c_{00}d_{00} + a_{10}b_{10} + c_{10}d_{10}$ for integers $a_{00}$, $b_{00}$, $c_{00}$, $d_{00}$, $a_{10}$, $b_{10}$, $c_{10}$, $d_{10}$.

Any target unitary $U$ must have this form due to~\eqnref{eq:nondeterministic-unitary-form}.
In other words, the \emph{only} unitaries that can be obtained by $\{\Clifford,T\}$ circuits of the form shown in~\figref{fig:rus-circuit} are those that can be expressed by entries in $\Aring$ after multiplying by a scalar.
Nonetheless, this restricted class can be used to approximate arbitrary unitaries more efficiently than unitaries limited to $\KMMring$, as we show in~\secref{sec:results} and~\secref{sec:applications}.

\subsection{Success probability and expected cost
\label{sec:rus.cost}
}
The success probability, i.e., the probability of obtaining the zero outcome for all ancilla measurements, can be computed from~\eqnref{eq:HT-unitarity-condition} and is given by
\begin{equation}
\Pr[\text{success}] = \frac{\alpha_0^2}{2^k} \leq \frac{\alpha_0^2}{2^{\lceil 2\log_2 \alpha_0 \rceil}}
\enspace ,
\end{equation}
where since $\alpha_0^2 < 2^k$, we may use $k \geq \lceil 2\log_2 \alpha_0 \rceil$.
The circuits in~\figref{fig:v3-circuits}, for example, each yield a value of $\alpha_0=\sqrt{5}$ and therefore a success probability of $5/8$.
If $U$ appears multiple times in~\eqnref{eq:nondeterministic-unitary-form}, then we have
\begin{equation}
\Pr[\text{success}] = \frac{t \alpha_0^2}{2^k} \leq \frac{t \alpha_0^2}{2^{\lceil \log_2 t \alpha_0^2 \rceil}}
\enspace ,
\end{equation}
where $t$ is the number of times that $U$ appears.
This upper bound can be made arbitrarily close to one for large enough $t$.

The expected number of repetitions required in order to achieve success is given by a geometric distribution with expectation value $1/p$, and variance $(1-p)/p^2$, where $p = \Pr[\text{success}]$.
If $C(W)$ is the cost of implementing the unitary $W$, then the expected cost of the RUS circuit is given by $C(W)/p$ with a variance of $C(W)(1-p)/p^2$.  The resources required to implement a $\{\Clifford,T\}$ fault-tolerant circuit are often dominated by the cost of implementing the $T$ gate.  We therefore define $C(W)$ as the number of $T$ gates in the circuit used to implement $W$.

The $T$-gate count is not the only reasonable cost function.  Other possibilities include circuit size, width, area or volume, or the total number of measurements.  The utility of a particular cost function varies depending on the target quantum computing architecture.  For architectures that use the surface code, for example, total volume can be a more complete metric than $T$ count~\cite{Fowler2013,Jones2013b}.

Here we choose to use $T$-gate count as the cost function because it is simple, and is consistent with other $\{\Clifford, T\}$-decomposition algorithms~\cite{Kliuchnikov2012,Amy2012,Selinger2012a,Kliuchnikov2012b,Wiebe2013,Gosset2013a,Ross2014}.
However, RUS circuits require techniques not present in the circuits produced by previous decomposition methods, such as rapid classical feedback and control, and active synchronization due to variable time scales per RUS circuit.
Thus, while $T$ count allows for direct comparison of RUS circuits with other methods, a more complete metric may be required in the future for resource calculations on a particular hardware architecture.

\subsection{Amplifying the success probability
\label{sec:rus.amplify}
}

The action of the multi-qubit unitary $W$ may be described by
\begin{equation}
\label{eq:rus-unitary-on-psi}
W\ket{0^m}\ket\psi = \sqrt{p}\ket{0^m}U\ket\psi + \sqrt{1-p}\ket{\Phi^\perp}
\enspace ,
\end{equation}
where $\ket{\Phi^\perp}$ is a state that depends on $\ket\psi$ and satisfies $(\ket{0^m}\bra{0^m}\otimes\Id)\ket{\Phi^\perp} = 0$.  That is, $W$ outputs a state which has amplitude $\sqrt{p}$ on the ``success'' subspace, and amplitude $\sqrt{1-p}$ on the ``failure'' subspace.  We show that in some cases we may apply amplitude amplification to boost the success probability and reduce the expected $T$ count of an RUS circuit.

Traditional amplitude amplification~\cite{Brassard2000} proceeds by applying the operator $(RS)^j$ on the initial state $W\ket{0^m}\ket\psi$ for some integer $j > 0$ and reflections
\begin{equation}\begin{split}
\label{eq:amp-amp-reflections}
S &= \Id - 2\ket{0^m}\ket\psi \bra{0^m}\bra{\psi}, \\
R &= WSW^\dagger = \Id - 2W\ket{0^m}\ket\psi \bra{0^m}\bra{\psi}W^\dagger
\enspace .
\end{split}\end{equation}
In the two-dimensional subspace spanned by $\{\ket{0^m}U\ket\psi, \ket{\Phi^\perp}\}$, $RS$ acts as a rotation by $2\theta$ where $\sin(\theta) = \sqrt{p}$.  Therefore $(RS)^j (W\ket{0^m}\ket\psi) = \sin((2j+1)\theta)\ket{0^m}U\ket\psi + \cos((2j+1)\theta)\ket{\Phi^\perp}$.  The goal then is to choose $j$ appropriately so as to minimize the expected number of $T$ gates.

The problem in this case is that $\ket\psi$ is unknown, and therefore we cannot directly implement $S$.
We can, however, implement
\begin{equation}
S' = \overline{\text{CZ}}(m) \otimes I
\enspace ,
\end{equation}
where $\overline{\text{CZ}}(m) = X^{\otimes m}\text{CZ}(m)X^{\otimes m}$ and $\text{CZ}(m)$ is the generalized controlled-$Z$ gate on $m$ qubits defined by
\begin{equation}
\text{CZ}(m) \ket{x_1,x_2,\ldots,x_m} = (-1)^{x_1 x_2\ldots x_m} \ket{x_1,x_2,\ldots,x_m}
\enspace .
\end{equation}
We could, therefore, apply $(WS'W^\dagger S')^j$ instead of $(RS)^j$.

In the case $m=1$ (one ancilla qubit) this procedure corresponds to so-called ``oblivious'' amplitude amplification.
\begin{lemma}[Oblivious amplitude amplification on $n+1$ qubits~\cite{Berry2013a}]
\label{lem:oblivious-amp-amp}
Consider a unitary $W$ that satisfies~\eqnref{eq:rus-unitary-on-psi} for $m=1$.
Let $S_1 := Z \otimes I$.  Then for any $j \in \mathbb{Z}$,
\begin{equation}
  (-W S_1 W S_1)^j W \ket{0}\ket\psi = \sin((2j+1)\theta) \ket{0}U\ket\psi + \cos((2j+1)\theta) \ket{1}\ket\phi
  \enspace ,
\end{equation}
where $\sin(\theta) = \sqrt{p}$.
\end{lemma}

In fact, oblivious amplitude amplification can be generalized to accommodate any number of ancilla qubits.
\begin{corollary}[Oblivious amplitude amplification on $n+m$ qubits]
\label{lem:amp-amp-cz}
Consider a unitary $W$ that satisfies~\eqnref{eq:rus-unitary-on-psi}.
Oblivious amplitude amplification on $\ket{0^m}U\ket\psi$ can be performed using the operator $WS'W^\dagger S'$, where $S' = \overline{\text{CZ}}(m)\otimes I$.
More precisely, for any $j \in \mathbb Z$
\begin{equation}
(-WS'W^\dagger S')^j (W\ket{0^m}\ket\psi) = \sin((2j+1)\theta)\ket{0^m}U\ket\psi + \cos((2j+1)\theta)\ket{\Phi^\perp}
\enspace ,
\end{equation}
where $\sin(\theta) = \sqrt{p}$.
\end{corollary}

\begin{proof}
The main technical part the proof of~\lemref{lem:oblivious-amp-amp} in~\cite{Berry2013a} is accomplished by another Lemma called the $2$D Subspace Lemma (see Lemma $3.6$ of~\cite{Berry2013a}).  Like~\lemref{lem:oblivious-amp-amp}, the $2$D Subspace Lemma is stated specifically for the $m=1$ case.  However, the proof still holds if $\ket 0$ is replaced by $\ket{0^m}$.  In that case, we find that the state 
\begin{equation}
\ket{\Psi^\perp} := W^\dagger \left(\sqrt{1-p}\ket{0^m}U\ket\psi - \sqrt{p}\ket{\Phi^\perp}\right)
\end{equation}
is both orthogonal to $\ket{0^m}\ket{\psi}$ and satisfies $(\ket{0^m}\bra{0^m}\otimes I)\ket{\Psi^\perp} = 0$.
This allows us to calculate the behavior of $W^\dagger$ within the two-dimensional subspace spanned by $\ket{0^m}U\ket\psi$ and $\ket{\Phi^\perp}$. We have
\begin{equation}
\begin{split}
W^\dagger (\ket{0^m}U\ket\psi) &= \sqrt{p}\ket{0^m}\ket\psi + \sqrt{1-p}\ket{\Psi^\perp} \\
W^\dagger \ket{\Phi^\perp} &= \sqrt{1-p}\ket{0^m}\ket\psi - \sqrt{p}\ket{\Psi^\perp}
\end{split}
\enspace .
\end{equation}

Just as in~\cite{Berry2013a}, this permits simple calculations yielding
\begin{equation}
-WS'W^\dagger S' (\ket{0^m}U\ket{\psi}) = \cos(2\theta)\ket{0^m}U\ket\psi + \sin(2\theta)\ket{\Phi^\perp}
\end{equation}
and
\begin{equation}
-WS'W^\dagger S \ket{\Psi^\perp} = \sin(2\theta)\ket{0^m}U\ket\psi + \cos(2\theta)\ket{\Phi^\perp}
\enspace .
\end{equation}
The conclusion is that $-WS'W^\dagger S'$ acts as a rotation by $2\theta$ in the two-dimensional subspace of interest.
\end{proof}

% \begin{lemma}[\cite{Kothari2013}]
% \label{lem:2d-subspace}
% Let $W$ be a unitary that satisfies~\eqnref{eq:rus-unitary-on-psi}. Then the state
% $$\ket{\Psi^\perp} := W^\dagger \left(\sqrt{1-p}\ket{0^m}U\ket\psi - \sqrt{p}\ket{\Phi^\perp}\right)$$
%  satisfies $(\ket{0^m}\bra{0^m}\otimes\Id)\ket{\Psi^\perp} = 0$.
% \end{lemma}

If $m \leq 2$, then $S'$ can be implemented with only Clifford gates, i.e., $X$ and either $Z$ or $\text{CZ}$.
Then, for a fixed value of $j$, the total number of $T$ gates in the corresponding amplified circuit is given by $(2j+1)T_0$.
In order for amplitude amplification to yield an improvement in the expected number of $T$ gates, we therefore require that
\begin{equation}
(2j+1)\sin^2(\theta) < \sin^2((2j+1)\theta)
\enspace ,
\end{equation}
a condition that holds if and only if $0 \leq p < 1/3$. Thus a sensible course of action is to apply amplitude amplification for all RUS circuits for which $p < 1/3$, and leave higher probability circuits unchanged.

Consider, for example, an RUS circuit that contains $15$ $T$ gates and has a success probability of $0.1$.  In this case, using amplitude amplification with a value of $j=1$ yields a new circuit with success probability $0.676$ and $45$ $T$ gates, an improvement in the expected number of $T$ gates by a factor of $2.25$.
The effects of amplitude amplification on our database of RUS circuits are discussed in~\secref{sec:results}.

Cost analysis of amplitude amplification for circuits with more than two ancilla qubits is more complicated because the reflection operator $S'=\overline{\text{CZ}}(m)$ is not a Clifford gate. For three ancilla qubits, for example, $S'$ requires the controlled-controlled-$Z$ gate, which can be implemented with $4$ $T$ gates~\cite{Jones2012d}.  Larger versions of $\text{CZ}(m)$ could be synthesized directly~\cite{Kliuchnikov2013a,Welch2013}, or by using a recursive procedure~\cite{Nielsen2000}.  The circuits presented in~\secref{sec:results} use at most two ancilla qubits, however, so more complicated amplification circuits are not an issue in our analysis.

\section{Direct search algorithm}
\label{sec:search}
While equations~\eqnref{eq:nondeterministic-unitary-form} and~\eqnref{eq:rus-scalar-condition} restrict the kinds of unitaries that can be exactly obtained with RUS circuits, they indicate very little about how to implement the multi-qubit unitary $W$.
Given $W$ explicitly, it is possible to synthesize a corresponding $\{\Clifford,T\}$ circuit with a minimum number of $T$ gates~\cite{Gosset2013a}, at least for $W$ with small $T$ count.
However, given a unitary $U$ of the form~\eqnref{eq:rus-unitary}, there are potentially many choices of $W$, and an efficient way to find the $W$ that will result in the minimum number of $T$ gates is unknown (and a direction for future research).

As a step towards synthesizing RUS circuits and understanding their scope, we design an optimized direct search algorithm that synthesizes RUS circuits up to a given $T$-gate count. Our direct search algorithm is as follows:
\begin{enumerate}
\itemsep 1pt
 \parskip 0pt
  \item Select the number of ancilla qubits and the number of gates.
  \item Construct a $\{\Clifford,T\}$ circuit and compute the resulting unitary matrix $W$.
  \item Partition the first two columns of $W$ into $2\times 2$ matrices.
  \item Identify and remove matrices that are proportional to Clifford gates.
  \item If the remaining matrices are all proportional to the same unitary matrix, then keep the corresponding circuit.
\end{enumerate}

We restrict the recovery operations ${R_i}$ of the circuits in our direct search to the set of single-qubit Cliffords.
This choice is motivated by our use of the $T$ count as a cost function; Clifford gates, and therefore the recovery operations are assigned a cost of zero, therefore such recovery operations are inexpensive.

In order to identify relevant search parameters for step 1 and circuit constructions for step 2, we initially performed a random search over a wide range of circuit widths (number of qubits) and sizes (number of gates).  Our search produced ample results for small numbers of ancilla qubits, large numbers of $T$ gates, and just one or two entangling gates.  We therefore focus our current study on circuits of the form shown in~\figref{fig:two-cz-canonical}, which contain one ancilla qubit and two CZ gates, interleaved with single-qubit Clifford gates.

Naively, the number of circuits of the form given in~\figref{fig:two-cz-canonical} is $O(3^n)$, where $n$ is the maximum number of (non-CZ) gates in the circuit, and the base of three is the size of the set $\{H,S,T\}$.
In order to reduce the time complexity of direct search, we constructed each single-qubit gate sequence using the canonical form proposed in~\cite{Bocharov2012}. A canonical form sequence is the product of three $2\times2$ unitary matrices $g_2Cg_1$ where $g_1,g_2$ belong to the single-qubit Clifford group, and $C$ is the product of some number of ``syllables'' $TH$ and $SHTH$.  The canonical form yields a unique representation of all single-qubit circuits over $\{H,T\}$; there are $2^{t-3}+4$ canonical circuits of $T$-count at most $t$.  The canonical representation yields more than a quadratic improvement in time complexity compared to naive search, since the number of $T$ gates is roughly one-half the total number of gates.

In general, the canonical form requires conjugation by the full single-qubit Clifford group, which contains $24$ elements.  Given a product of syllables $C$, each of the $24^2=576$ circuits $g_2Cg_1$ are unique.  However, when multiple canonical form circuits are composed in a larger circuit, as in~\figref{fig:two-cz-canonical}, some combinations of Clifford gates can be eliminated. For example, when $g_2Cg_1$ is applied to the state $\ket{0}$, $g_1$ need only be an element of $\{I,X,SH,SHX,HSH,HSHX\}$ since diagonal gates act trivially on $\ket{0}$.  Similar simplifications for~\figref{fig:two-cz-canonical} are shown in~\figref{fig:canonical-simplifications}.  In total, these Clifford optimizations further reduce the search space by a factor of more than $10^5$.

\begin{figure}
\centering
\includegraphics{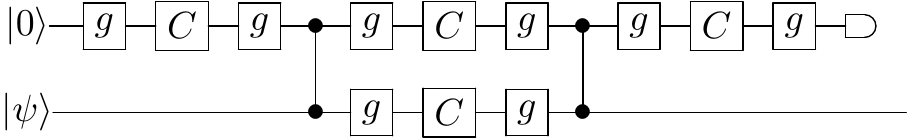}
\caption{\label{fig:two-cz-canonical}
The general form of most RUS circuits in our database.  Each of the gates labeled $g$ represents an element of the single-qubit Clifford group.  Each of the gates labeled $C$ represents a single-qubit canonical circuit as defined in~\cite{Bocharov2012}.
}
\vspace{.8cm}
\begin{subfigure}[b]{.49\textwidth}
\centering
\includegraphics[scale=1]{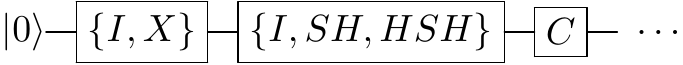}
\caption{}
\end{subfigure}
\begin{subfigure}[b]{.49\textwidth}
\centering
\includegraphics[scale=1]{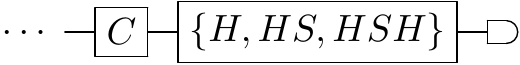}
\caption{}
\end{subfigure}
\begin{subfigure}[b]{\textwidth}
\hfill
\end{subfigure}
\begin{subfigure}[b]{.49\textwidth}
\centering
\includegraphics[scale=1]{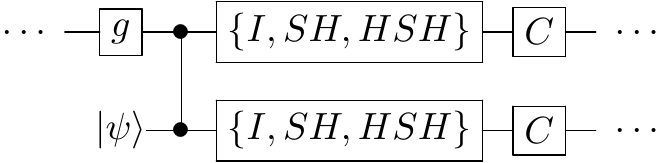}
\caption{}
\end{subfigure}
\begin{subfigure}[b]{.49\textwidth}
\centering
\includegraphics[scale=1]{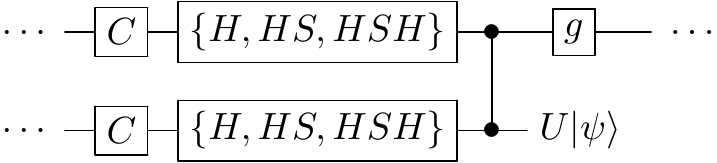}
\caption{}
\end{subfigure}
\caption{\label{fig:canonical-simplifications}
Some gates $g$ in~\figref{fig:two-cz-canonical} can be restricted to a subset of the single-qubit Clifford group. (a) Circuits that begin with diagonal gates can be eliminated since they add a trivial phase to $\ket 0$. (b) Similarly, diagonal gates have no impact on the $Z$-basis measurement. (c) Pauli gates and $S$ gates can be commuted through the CZ and absorbed into either $\ket\psi$ or the preceding $g$ gate. (d) Analogously, Pauli and $S$ gates occurring before the CZ can be absorbed by the trailing $g$ gate or by the output.
}
\end{figure}

Despite these optimizations, our direct search algorithm still requires time exponential in the number of $T$ gates.  To further reduce the time complexity, we partitioned the search into thousands of small computations running in parallel on a large cluster and collected the results in a central database.  We were able to exhaustively synthesize circuits of the form given in~\figref{fig:two-cz-canonical} up to a total (raw) $T$ count of $15$ in roughly one week running on hundreds of cores.
The results of our direct search algorithm are presented in the next section.

\section{Direct search results
\label{sec:results}
}

Our search yielded many RUS circuits that implement the same unitary $U$, but with different $T$-gate counts and success probabilities. To eliminate redundancy we construct a database containing only the circuit with the minimum expected $T$ count for a given unitary $U$.
The resulting database contains $2194$ RUS circuits each of which contains at most $15$ $T$ gates.  Upon success, each circuit exactly implements a unique non-Clifford single-qubit unitary $U$, and otherwise implements a single-qubit Clifford operation.
The database statistics are shown in~\figref{fig:search-results}. For circuits with success probability less than $1/3$, we used amplitude amplification to improve performance (see~\secref{sec:rus.amplify}).
Most RUS circuits result in high success probability and low expected $T$ count.
\figref{fig:histogram-expected-tcount} illustrates the impact of amplitude amplification on the expected $T$ count.  Amplification improved the performance of circuits with relatively high expected $T$ count, but did not improve circuits with expected $T$ count of $30$ or less.
In general, RUS circuits exhibit very low expected $T$ counts around $15$--$20$.
Note that the database also includes some circuits that were found by preliminary searches not of the form of~\figref{fig:two-cz-canonical}.

\begin{figure}[tb]
\centering
\begin{subfigure}[b]{.49\textwidth}
\centering
\includegraphics[height=7.71cm]{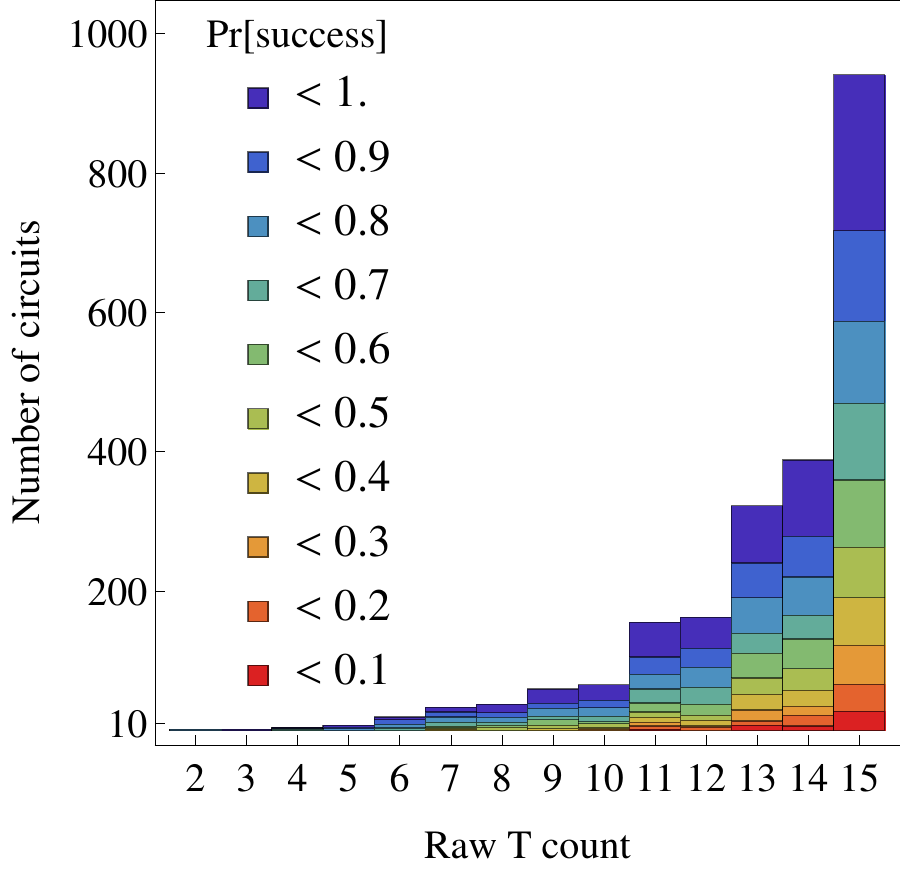}
\caption{\label{fig:histogram-raw-tcount}
}
\end{subfigure}
\hfill
\begin{subfigure}[b]{.48\textwidth}
\centering
\includegraphics[height=7.71cm]{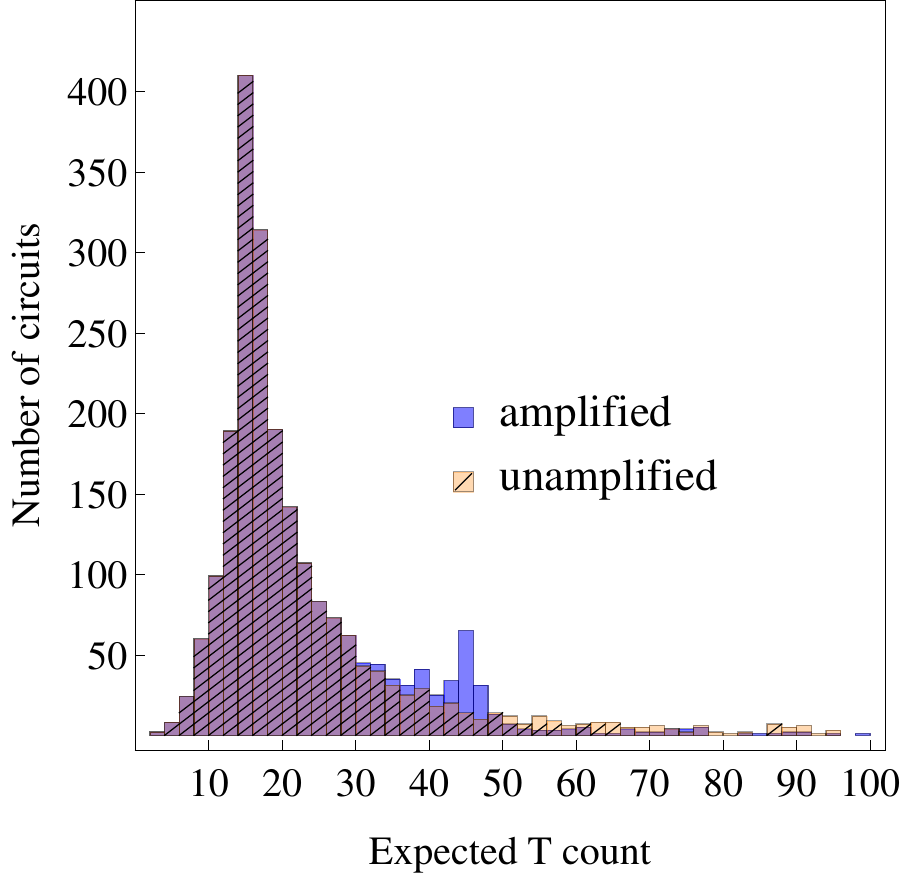}
\caption{\label{fig:histogram-expected-tcount}
}
\end{subfigure}
\caption{\label{fig:search-results}
Statistics for the database of repeat-until-success circuits, including all circuits of the form of~\figref{fig:two-cz-canonical} up to a $T$ count of $15$.  (a) The total number of circuits grouped by (raw) $T$ gate count and success probability. (b) The total number of circuits grouped by expected $T$ count, both before amplitude amplification and after amplitude amplification.  The two histograms (before amplification and after amplification) are overlayed, where the darker hatched bars indicate circuits that are unaffected by amplification.  Only circuits with an expected $T$ count of at most $100$ are shown.
}
\end{figure}

Of the $2194$ RUS circuits, $1659$ are axial rotations, i.e., unitaries which, modulo conjugation by Cliffords, are rotations about the $Z$-axis of the Bloch sphere, and $535$ are non-axial rotations.
The number of axial rotations is noteworthy since, modulo Clifford conjugation, only one non-trivial single-qubit rotation can be exactly synthesized with $\{\Clifford,T\}$ and without measurement, namely $T$~\cite{Kliuchnikov2012}.  Our results show that \emph{many} axial rotations can be implemented exactly (conditioned on success) when measurement is allowed.

Remarkably, the non-axial rotations in our database offer an expected $T$ count that is dramatically better than the $T$ count obtained by approximation algorithms~\cite{Selinger2012a,Kliuchnikov2012b,Ross2014}. For each RUS circuit in the database we computed the number of $T$ gates required to approximate the corresponding unitary to within a distance of $10^{-6}$ using the algorithm of KMM.  \figref{fig:kmm-ratios} shows the ratio of the $T$ count given by KMM vs.~the expected $T$ count for the RUS circuit. (KMM and Ross-Selinger achieve similar $T$ count scaling so we expect similar ratios when comparing to Ross-Selinger.) Our results show a typical improvement of about a factor of three for axial rotations and a typical improvement of about a factor of about $12$ for non-axial rotations.
The larger improvement for non-axial rotations is expected since the KMM algorithm requires the unitary to be first decomposed into a sequence of three axial rotations.

\begin{figure}[tb]
\centering
\includegraphics[scale=1]{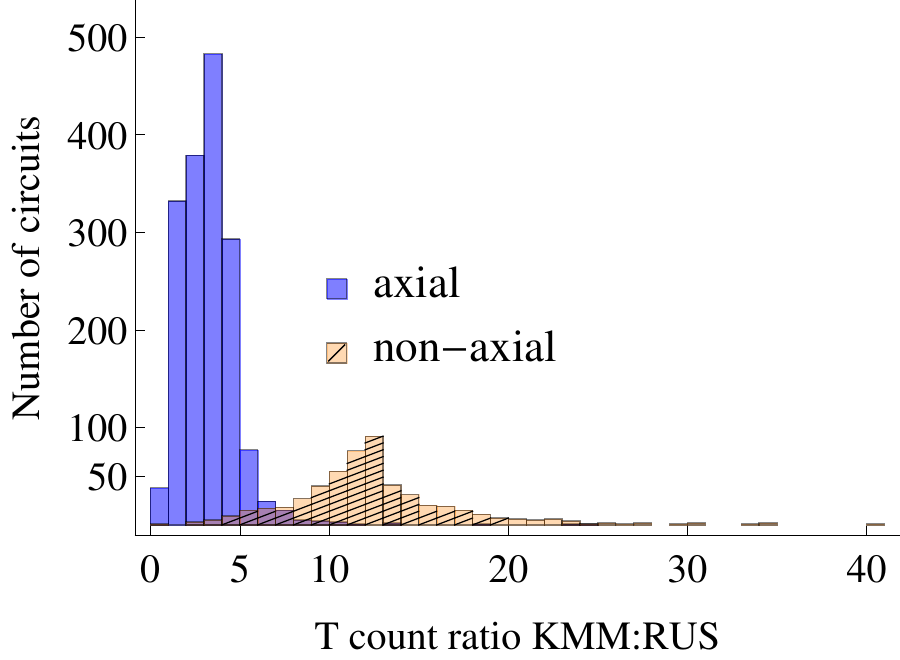}
\caption{\label{fig:kmm-ratios}
RUS circuits database split into axial and non-axial single-qubit rotations.  For each circuit, the number of $T$ gates required to approximate the corresponding ``success'' unitary $U$ to within $10^{-6}$ was calculated using the algorithm of~\cite{Kliuchnikov2012b}.  The $x$-axis represents the ratio of the KMM $T$ count vs.~the expected number of $T$ gates for the RUS circuit.
}
\end{figure}

As an example, the RUS circuit shown in~\figref{fig:large-kmm-ratio-circuit} implements the non-axial single-qubit rotation $U = (2X + \sqrt{2}Y + Z)/\sqrt{7}$ with four $T$ gates and a probability of success of $7/8$.  By contrast, approximating $U$ to within $\epsilon = 10^{-6}$ using the KMM algorithm requires a total of $182$ $T$ gates.  Thus the circuit in~\figref{fig:large-kmm-ratio-circuit} not only implements the intended unitary exactly, but does so at a cost over $40$ times less than the best approximation methods.

\begin{figure}
\centering
\includegraphics{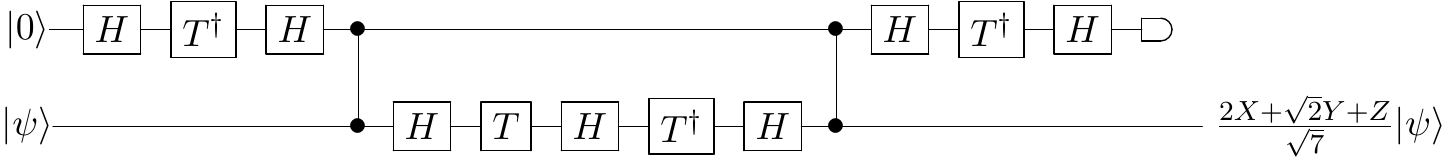}
\caption{\label{fig:large-kmm-ratio-circuit}
An RUS circuit to implement the unitary $U = (2X + \sqrt{2}Y + Z)/\sqrt{7}$ with probability $7/8$, and $Z$ otherwise.  Approximation of $U$ without ancillas requires $182$ $T$ gates (roughly $40$ times more) for $\epsilon = 10^{-6}$.
}
\end{figure}

Our database is too large to offer an analysis of each circuit in detail.
However, we highlight some particularly important examples.
The smallest circuit in our database contains two $T$ gates and is shown in~\figref{fig:gosset-unitary}. Upon measuring zero, which occurs with probability $3/4$, the circuit implements $(\Id + i\sqrt{2}X)/\sqrt{3}$ and upon measuring one implements $\Id$.
This circuit was predicted to exist by Gosset and Nagaj~\cite{Gosset2013}. They required a $\{\Clifford,T\}$ circuit that exactly implemented $R=(\sqrt{2}\Id-iY)/\sqrt{3}$ with a constant probability of success.  The unitary implemented by~\figref{fig:gosset-unitary} is equivalent to $R$ up to conjugation by Clifford gates.
% The conjugation relation is: R = SX (\sqrt{2}X - \Id) S^\dagger /\sqrt{3}.

\begin{figure}
\centering
\includegraphics{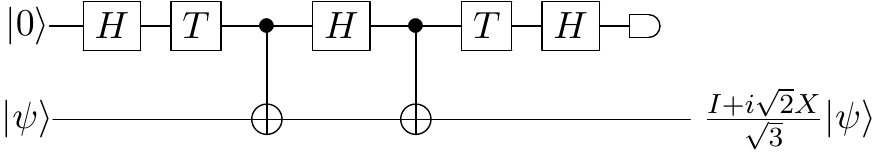}
\caption{\label{fig:gosset-unitary}
The smallest circuit in our database.  Upon measuring zero, with probability $3/4$, it implements $(\Id + i\sqrt{2}X)/\sqrt{3}$ on the input state $\ket\psi$.  Upon measuring one, it implements the identity.
}
\end{figure}

As discussed in~\secref{sec:introduction}, our database contains a circuit that implements $V_3$.  In addition to the circuit shown in~\figref{fig:v3-5.26}, our search also found a circuit that implements $V_3$ with the same number of $T$ gates, but with just a single ancilla qubit, as shown in~\figref{fig:v3-one-ancilla}.  The expected $T$ count of the single-ancilla circuit is slightly worse than that of~\figref{fig:v3-5.26}, though, since all four of the $T$ gates on the ancilla must be performed ``online''.

\begin{figure}
\centering
\includegraphics{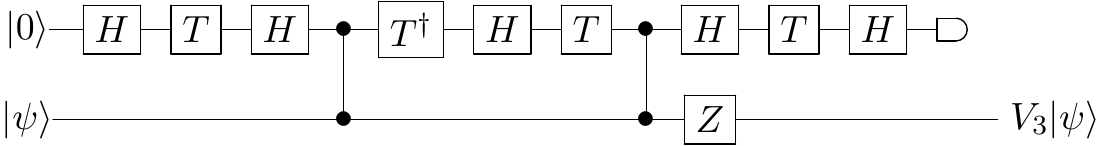}
\caption{A circuit, like the circuits in \figref{fig:v3-circuits}, to implement $V_3$ with probability $5/8$ and identity with probability $3/8$, using only one ancilla qubit and one measurement.
\label{fig:v3-one-ancilla}
}
\end{figure}

The $V_3$ gate is one of a family of $V$-basis gates for which the normalization factor is $1/\sqrt{5}$.
In addition to single-qubit unitary decomposition based on $V_3$, \cite{Bocharov2013} also offers the possibility of decomposing single-qubit unitaries using $V$-basis gates with normalization factors $1/\sqrt{p}$ where $p$ is a prime.  These ``higher-order'' $V$ gates cover $SU(2)$ more rapidly than $V_3$ and therefore offer potentially more efficient decomposition algorithms.  A number of such $V$-basis gates can be found in our database, including axial versions for $p\in \{13,17,29\}$, as shown in~\figref{fig:high-order-v-circuits}, offering the first fault-tolerant implementations of these gates.  The prospect of decomposition algorithms with these circuits is discussed in~\secref{sec:applications.v3}.

\begin{figure}
\begin{subfigure}[b]{\textwidth}
\includegraphics[width=\textwidth]{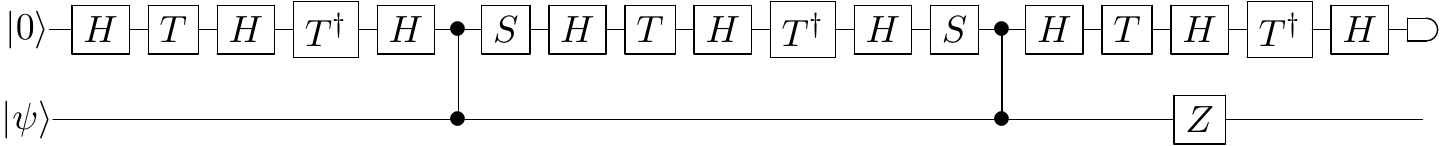}
\caption{$(3I+2iZ)/\sqrt{13}$, $\Pr=13/16$}
\vspace{.2cm}
\end{subfigure}
\begin{subfigure}[b]{\textwidth}
\includegraphics[width=\textwidth]{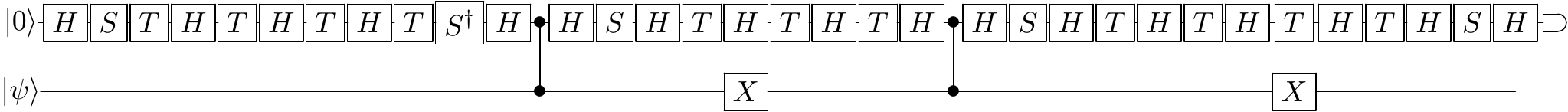}
\caption{$(4I+iZ)/\sqrt{17}$, $\Pr \approx 0.985$}
\vspace{.2cm}
\end{subfigure}
\begin{subfigure}[b]{\textwidth}
\includegraphics[width=\textwidth]{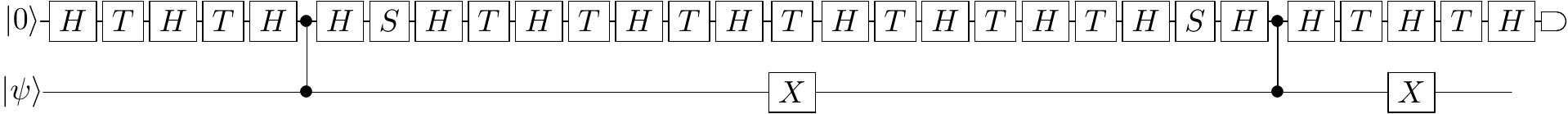}
\caption{$(5I+2iZ)/\sqrt{29}$, $\Pr \approx 0.774$}
\end{subfigure}
\caption{\label{fig:high-order-v-circuits}
RUS circuits for $V$-basis gates with prime normalization factors (a) $p = 13$ (b) $p=17$ and (c) $p = 29$.  The values under each circuit indicate the unitary effected upon success and the success probability, respectively.  Each circuit implements the identity upon failure.
}
\end{figure}

\section{Applications
\label{sec:applications}
}

One application of RUS circuits is in the construction of universal sets of gates.
Our RUS circuits offer exact, fault-tolerant implementations of a large set of single-qubit unitary gates.
The Clifford group plus any one non-Clifford gate is universal for quantum computation~(see, e.g., \cite{Campbell2012} Appendix D). Thus any of our RUS circuits can be used to construct a new universal gate set.  The question, though, is whether or not RUS circuits can be used to decrease resource costs of unitary approximation methods.

In this section, we show that RUS circuits can be used to significantly improve upon approximate decomposition of single-qubit unitaries.  First we discuss the use of our improved $V_3$ circuit for decomposition into $\{\Clifford,V_3\}$.  Then we show how to compose RUS circuits in series in order to expand the size and density of the database.  The expanded database can be used to approximate single-qubit unitaries up to an accuracy that is sufficient for a number of important quantum algorithms.  In particular, in~\secref{sec:applications.phase-estimation}, we show how to use circuits in our database for applications using the quantum phase estimation algorithm.

\subsection{Decomposition with \texorpdfstring{$V_3$}{V3}
\label{sec:applications.v3}
}
The RUS circuit for $V_3$, shown in~\figref{fig:v3-5.26}, can be used directly in the decomposition algorithm of~\cite{Bocharov2013}.  The BGS algorithm produces an $\epsilon$-approximation of a given single-qubit unitary with $3\log_5(1/\epsilon)$ $V_3$ gates in most cases.  Multiplying by an expected $T$-cost of $5.26$, using the circuit in~\figref{fig:v3-5.26}, yields an algorithm with an expected $T$ count of
\begin{equation}
\label{eq:bgs-tcount}
15.78\log_5(1/\epsilon)
\enspace .
\end{equation}
This is an improvement over the estimated $T$ count of $3(3.21\log_2(3/\epsilon)-6.93)$~\cite{Kliuchnikov2012b} for all $\epsilon < 0.25$.

% Fault-tolerant implementations provide logical gates that are accurate up to an arbitrarily small error probability $p$.  The total error probability of a fault-tolerant quantum circuit can be calculated using a union bound.  Thus a circuit with $n$ gates will require a fault-tolerant implementation of each gate that has error probability roughly $1/n$.  Using our implementation of $V_3$, therefore, reduces the number of $T$ gates required to approximate a single qubit unitary \emph{and} reduces the necessary resource overhead for each $T$ gate (in expectation).

% Fault-tolerance resource costs of can be estimated by the number of ``physical'' $T$ gates. \figref{fig:phys-tcount-rus-kmm} shows the physical $T$ gate count of $V$-basis and KMM decomposition (assuming a scaling of $9.63\log_2(1/\epsilon)-20.79$).  Both decomposition methods are competitive, but the advantage of the $V$-basis decomposition increases with the rotation accuracy. \todo{Maybe try some alternate values of $\epsilon/p$.}

% \begin{figure}
% \centering
% \includegraphics[width=.5\textwidth]{images/phys-tcount-rus-kmm}
% \caption{\label{fig:phys-tcount-rus-kmm}
% }
% \end{figure}

The database also contains $V$-basis gates with prime normalization factors larger than $5$. In~\cite{Bocharov2013}, the authors conjecture that the decomposition algorithm for $p=5$ extends to other primes with a $T$-count scaling of
$4\log_p(1/\epsilon)$.
However, whereas $p=5$ requires only the single $V_3$ gate, higher prime values require implementation of multiple $V$ gates. For simplicity, assume that each of the required $V$ gates can be implemented with $T$-count $T_p$.  Then the decomposition achieved for prime $p$ will be better than that obtained with $V_3$ if
\begin{equation}
\label{eq:high-order-V-ratio}
1 < \frac{5.26}{T_p} \log_5(p)
\enspace .
\end{equation}

Unfortunately, our database contains only a single $V$-basis gate for each of $p = \{13,17,29\}$.  For the sake of argument, we calculate~\eqnref{eq:high-order-V-ratio} under the optimistic assumption that for each $p$, the remaining $V$ gates can someday be implemented at the same cost $T_p$.  Using the circuits in~\figref{fig:high-order-v-circuits} we obtain
\begin{subequations}
\begin{align}
&5.26/7.38 \log_5{13} \approx 1.13,\\
&5.26/11.17 \log_5{17} \approx 0.83,\\
&5.26/14.22 \log_5{17} \approx 0.77
\enspace .
\end{align}
\end{subequations}
Based on these calculations we conclude that, while improved decomposition may be possible using $p=13$, higher values of $p$ are unlikely to yield cost benefits on their own.

On the other hand, given implementations of multiple $V$ gates, there is no reason to limit to a single value of $p$.  One could imagine an algorithm that combined multiple classes of $V$ gates, using largely $V_3$ and using more expensive high-order $V$ gates selectively.  We do not consider such an algorithm directly.  In the next section, however, we study the effect of optimally combining all of the RUS circuits in our database, not just $V$ gates.

\subsection{Decomposition by composition of RUS circuits
\label{sec:applications.database}
}
% The $V_3$ circuit offers an algorithm for improved decomposition of single-qubit unitaries by alternating $V_3$ with Clifford gates.  A more dramatic improvement can be obtained by exploiting a larger number of circuits from our database.

It is possible to approximate a given single-qubit unitary $U$ to within any $\epsilon$ by composing Clifford gates and circuits from our database.  But finding the optimal composition sequence among all possible compositions of circuits is a challenging task.  Ideally, we could construct an efficient decomposition algorithm based on algebraic characterization of the set of RUS circuits, similar to algorithms for other gate sets~\cite{Selinger2012a,Kliuchnikov2012b,Bocharov2013,Ross2014}.  But the current theoretical characterization of RUS circuits remains open is a direction for future work.
Here, we develop decomposition algorithm based on exhaustive composition of RUS circuits, which is similar in nature to the methods of~\cite{Fowl04c} and~\cite{Bocharov2012}.

Starting with the set of RUS circuits found by our direct search algorithm, we compute all products of pairs of circuits, keeping those that produce a unitary which is not yet in the database.  Composite circuits of arbitrary size can be constructed in this manner: triples of circuits can be constructed from singles and pairs, and so on.  Call a circuit a class-$k$ circuit if it is composed of a $k$-tuple of RUS circuits from the original database.  Then the number $N_k$ of class-$k$ circuits is bounded by
\begin{equation}
N_k \leq N_1 \cdot N_{k-1} \leq N_1^k
\enspace ,
\end{equation}
where $N_1$ is the number of circuits in the original database.

To manage the database expansion, we keep only those circuits that yield an expected $T$ count of at most some fixed value $T_0$.  This has the simultaneous effect of discarding poorly performing circuits and reducing the value of $N_k$ so that construction of class-$(k+1)$ circuits is less computationally expensive.  Furthermore, circuits can be partitioned into equivalence classes by Clifford conjugation.  The unitaries of the initial set of circuits are of the form $g_0 U g_1$, where $U$ is the unitary obtained from the RUS circuit, and $g_0, g_1$ are single-qubit Cliffords.  Thus, the product of $k$ such circuits has the form
\begin{equation}
g_0 U_1 g_1 U_2 g_2\ldots U_k g_{k}
\enspace .
\end{equation}
The set of class-(${k+1}$) circuits can then be constructed by using
\begin{equation}
g_0 U_1 g_1 U_2 g_2\ldots U_k g_{k} (g_{k'} U_{k+1} g_{k+1}) = g_0 U_1 g_1 U_2 g_2\ldots U_k g_{k''} U_{k+1} g_{k+1}
\enspace,
\end{equation}
so that the Clifford $g_k$ is unnecessary.  Furthermore, $g_0$ can always be prepended later, and so we instead express each class-$k$ unitary as
\begin{equation}
\label{eq:clifford-equiv-representative}
U_1 g_1 U_2 g_2\ldots U_k
\enspace .
\end{equation}

To find an equivalence class representative of $U$, we first remove the global phase by multiplying by $u^*/\sqrt{|u|^2}$, where $u$ is the first non-zero entry in the first row of $U$. Next, we conjugate $U$ by all possible pairs of single-qubit Cliffords. The first element of a lexicographical sort then yields the representative $g_1 U g_2$ for some Cliffords $g_1, g_2$.

Once the expanded database has been constructed up to a desired size, the decomposition algorithm is straightforward.  Given a single-qubit unitary $U$ and $\epsilon \in [0,1]$, select all database entries $V$ such that $D(U,V) \leq \epsilon$, where
\begin{equation}
\label{eq:fowler-distance}
D(U,V) = \sqrt{\frac{2-\abs{\Tr(U^\dagger V)}}{2}}
\end{equation}
is the distance metric defined by~\cite{Fowl04c} and also used by~\cite{Selinger2012a, Kliuchnikov2012b,Bocharov2013,Wiebe2013,Ross2014}.
Then, among the selected entries, find and output the circuit with the lowest expected $T$ count.

\subsubsection{Results: decomposition with axial rotations
\label{sec:applications.database.axial}
}

An arbitrary single-qubit unitary can be decomposed into a sequence of three $Z$-axis rotations and two Hadamard gates~\cite{Nielsen2000}.  Therefore, approximate decomposition of $Z$-axis rotations suffices to approximate any single-qubit unitary.
If we limit to $Z$-axis, i.e, diagonal, rotations only, then a few additional simplifications are possible.  In particular, each unitary can be represented by a single real number corresponding to the rotation angle in radians.  The result of a sequence of such rotations is then given by the sum of the angles.  Furthermore, up to conjugation by $\{X,S\}$, all $Z$-axis rotations can be represented by an angle in the range $[0, \pi/4]$.  This allows for construction of a database of $Z$-axis rotations which is much larger than a database of arbitrary (non-axial) unitaries.

Using the database expansion procedure described above, we construct a database containing all combinations of RUS circuits with expected $T$ count at most $30$. The maximum distance (according to~\eqnref{eq:fowler-distance}) between any two neighboring rotations is less than $2.8\times 10^{-6}$, and can be improved to $2\times 10^{-6}$ by selectively filling the largest gaps. So the resulting database permits approximation of any $Z$-axis rotation to within $\epsilon = 10^{-6}$.

To approximate a $Z$-axis rotation by an angle $\theta$, we select all entries that are within the prescribed distance $\epsilon$, and then choose the one with the smallest expected $T$ count.  This procedure is efficient since the database can be sorted according to rotation angle.  Then the subset of entries that are within $\epsilon$ can be identified by binary search.

In order to assess the performance of this method, we approximate, for various values of $\epsilon$, a sample of $10^5$ randomly generated angles in the range $[0,\pi/4]$.
%(Any $Z$-axis rotation can be expressed as a rotation by $\theta = \theta_1 + \theta_2$, where $\theta_1 \in [0,\pi/4]$ and $\theta_2 \in {\pm \pi/2, \pi}$ can be implemented by a Clifford.)
Results are shown in~\figref{fig:rus-axial-decomposition} and~\tabref{tab:rus-axial-decomposition}.  A fit of the mean expected $T$ count for each $\epsilon$ yields a scaling given by~\eqnref{eq:axial-tcount-scaling}, with a slope roughly $2.4$ times smaller than that reported by~\cite{Kliuchnikov2012b} for the rotation $R_Z(1/10)$.

\begin{table}[t]
    \begin{minipage}[b]{.57\textwidth}
        \includegraphics[width=.99\textwidth]{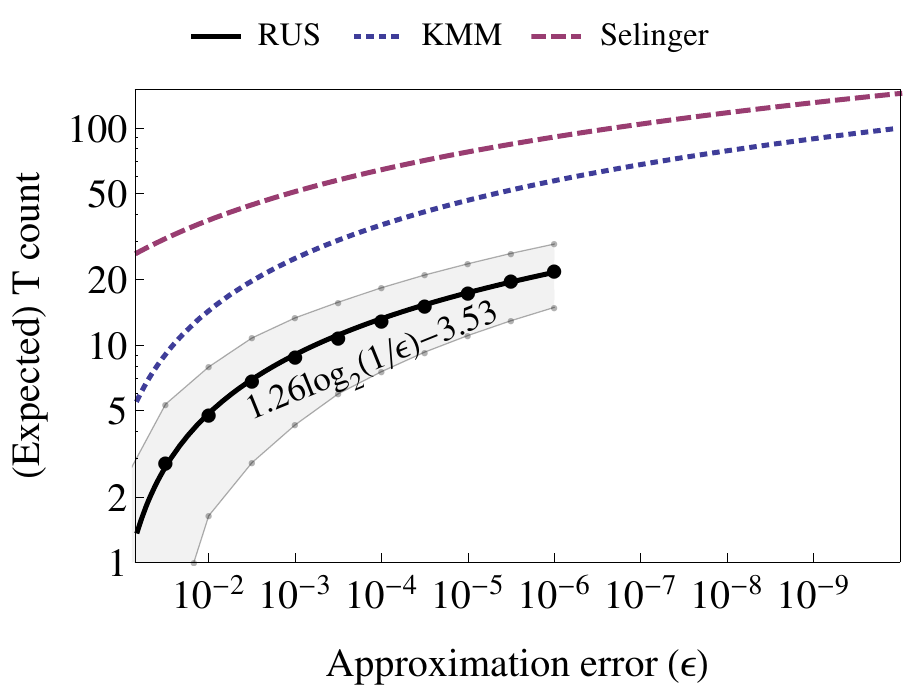}
        \captionof{figure}{The expected number of $T$ gates required to approximate a single-qubit $Z$-axis rotation to within a distance $\epsilon$ over $10^5$ real numbers selected in the range $[0,\pi/4]$ uniformly at random.  For each value $\theta$, the RUS circuit with the smallest expected $T$ count within $\epsilon$ of the unitary $R_Z(\theta)$ was selected.  The mean for each value of $\epsilon$ is plotted, yielding a fit-curve of $1.26 \log_2(1/\epsilon) - 3.53$.  The gray region is an estimate of the interval containing the actual number of $T$ gates with probability $95\%$. The other curves are included for reference: KMM $=3.21\log_2(1/\epsilon)-6.93$~\cite{Kliuchnikov2012b}, Selinger $=4\log_2(1/\epsilon) + 11$~\cite{Selinger2012a}.}
        \label{fig:rus-axial-decomposition}
    \end{minipage}
    \hfill
    \begin{minipage}[b]{.4\textwidth}
        \centering
        \begin{tabular}{@{}l@{~~}l@{~~}l@{}}
			$\log_{10}(1/\epsilon)$ & Exp $T$ ($\sigma^2$)& $\pm 95\%$ ($\sigma^2$)\\
			\hline
			$1$  &$1.1$ ($1.1$) &$1.2$ ($3.6$)\\
			$1.5$&$2.9$ ($2.2$) &$2.5$ ($2.9$)\\
			$2$  &$4.8$ ($3.4$) &$3.1$ ($2.9$)\\
			$2.5$&$6.8$ ($3.9$) &$4.0$ ($3.8$)\\
			$3$  &$8.8$ ($4.3$) &$4.5$ ($4.7$)\\
			$3.5$&$10.9$ ($4.6$)&$4.9$ ($5.2$)\\
			$4$  &$12.9$ ($4.8$)&$5.4$ ($5.5$)\\
			$4.5$&$15.1$ ($5.3$)&$5.9$ ($5.7$)\\
			$5$  &$17.4$ ($5.7$)&$6.3$ ($5.8$)\\
			$5.5$&$19.6$ ($6.0$)&$6.7$ ($6.1$)\\
			$6$  &$22.0$ ($6.4$)&$7.1$ ($6.5$)
		\end{tabular}
		\vspace{.5cm}
        \caption{Expected $T$ count required to approximate a random single-qubit $Z$-axis rotation with an RUS circuit.  The middle column indicates the expected $T$ count based on a sample of $10^5$ random angles.  The right-hand column indicates the expected $95$ percent confidence interval of the $T$ count for the best RUS circuit, given a random angle $\theta$.  The variance of each expected value is indicated in parenthesis.}
        \label{tab:rus-axial-decomposition}
        \vspace{.5cm}
    \end{minipage}
\end{table}

By way of comparison, Wiebe and Kliuchnikov report a scaling of $1.14\log_2(1/\theta)$ for small angles $\theta$.  However, their RUS circuits are specially designed for small angles.  For arbitrary angles they report an expected $T$ count of about
\begin{equation}
\label{eq:wiebe-kliuchnikov-tcount}
1.14\log_2(10^\gamma) + 8\log_2(10^{-\gamma}/\epsilon)
\enspace,
\end{equation}
where $\theta = a\times 10^{-\gamma}$ for some $a \in (0,1)$ and integer $\gamma > 0$.
Using~\eqnref{eq:wiebe-kliuchnikov-tcount} to calculate costs for the same $10^5$ random angles as above, we obtain a fit function of
\begin{equation}
\label{eq:wiebe-kliuchnikov-coarse-fit}
6\log_2(1/\epsilon) - 2.2
\enspace .
\end{equation}
Equation \eqnref{eq:wiebe-kliuchnikov-coarse-fit} indicates that the efficiency of the circuits in~\cite{Wiebe2013} does not extend to coarse angles.  Nevertheless, in~\secref{sec:applications.database.hybrid} we show how to combine the circuits of Wiebe and Kliuchnikov with our RUS circuits to achieve good cost scaling for relatively high accuracies.

Equation~\eqnref{eq:axial-tcount-scaling} also implies that RUS $Z$-axis rotations can be used to approximate \emph{arbitrary} single-qubit unitaries with a scaling approaching that of optimal ancilla-free decomposition.  Since an arbitrary unitary can be expressed as a product of three axial rotations, the expected $T$ count for approximating an arbitrary single-qubit unitary is given by
$3.9\log_2(3/\epsilon) - 8.37$.
On the other hand, Fowler calculates an optimal $T$-count of $2.95\log_2(1/\epsilon) + 3.75$ (on average) without using ancillas~\cite{Fowl04c}.

Since our circuits are non-deterministic, we are also concerned with the probability distribution of the number of $T$ gates.  For each composite circuit in the database, we calculate the variance $\sigma^2$ of the $T$ count based on the variance of each individual circuit.  We then obtain a confidence interval using Chevyshev's inequality
\begin{equation}
\Pr(\abs{\text{Actual}[T] - \text{Exp}[T]} \geq k\sigma) \leq \frac{1}{k^2}
\enspace .
\end{equation}
~\tabref{tab:rus-axial-decomposition} shows the mean expected $T$ count for each $\epsilon$.  By also calculating the mean variance $\sigma^2$, we obtain an estimate of the corresponding $95\%$ confidence interval, shown by the gray region in~\figref{fig:rus-axial-decomposition}.  That is, for a randomly chosen angle $\theta$, the actual number of $T$ gates required to implement $R_Z(\theta)$ is within the given interval around $1.26\log_2(1/\epsilon)-3.53$, with probability $0.95$.

The approximation accuracy permitted by our database is limited by computation time and memory.  To maximize efficiency, we used floating-point (accurate to $14$ digits) rather than symbolic arithmetic. Construction of all RUS circuit combinations up to expected $T$ count of $30$ took roughly $20$ hours and $41$ GB of memory using Mathematica.  \tabref{tab:z-rotation-density} shows the number of circuit combinations and corresponding rotation angle densities for increasing values of the expected $T$ count.  The size and density of the database increases by roughly one order of magnitude for every five $T$ gates.  We expect that with a more efficient implementation---in C/C++ for example---the worst-case approximation accuracy could be improved.

\begin{table}
\centering
\begin{tabular}{llll}
\hline\hline
Max. exp. &      &               &\\
$T$ count & Size & Mean $D$ & Max $D$ \\
\hline
$5$  & $7$ & $0.04$ & $0.08$\\
$10$ & $134$ & $0.0021$ & $0.0066$\\
$15$ & $2079$ & $0.00013$ & $0.0014$\\
$20$ & $27420$ & $0.00001$ & $0.00017$\\
$25$ & $320736$ & $0.0000009$ & $0.000016$\\
$30$ & $3446708$ & $0.00000008$ & $0.0000028$\\
%$31$ & $4235018$ & $0.00000007$ & $0.0000025$
\hline\hline
\end{tabular}
\caption{\label{tab:z-rotation-density}
Size and density of the $Z$-axis rotation database according to the maximum expected number of $T$ gates.  The mean and the maximum distances between nearest neighbors is given in columns three and four, respectively.
}
\end{table}

\subsubsection{More accurate axial rotations using gearbox circuits
\label{sec:applications.database.hybrid}
}
The approximation accuracy of $Z$-axis rotations can be improved indirectly by combining our database of axial rotations with the floating-point approach of Wiebe and Kliuchnikov~\cite{Wiebe2013}.
In their approach, a $Z$-axis rotation by angle $\phi = a\times 10^{-\gamma}$ is approximated with a ``gearbox'' circuit that multiplies the mantissa $a \in (0,1)$ by the value $10^{-\gamma}$.  The $T$ count of the gearbox circuit scales as
\begin{equation}
\label{eq:gearbox-scaling}
\text{Exp}_Z^{\text{WK}}[T] = 2T(a,10^\gamma \epsilon) + 1.14\log_2(10^\gamma) + 12.2
\enspace ,
\end{equation}
where $T(a,\epsilon)$ is the number of $T$ gates required to approximate $R_Z(a)$ to within a distance $\epsilon$.  In~\cite{Wiebe2013}, Selinger's algorithm is used to approximate the mantissa $a$.  However, any approximation method may be used.

The gearbox circuits are most useful when the angle $\phi$ is very small, and the number of significant digits $m = \log_{10}(10^{-\gamma}/\epsilon)$ is also small.  In that case, \eqnref{eq:gearbox-scaling} is largely determined by the $1.14\log_2(10^\gamma)$ term, which scales better than any other known methods. The scaling is maintained even for very high accuracy, so long as the required relative precision is low.

If our decomposition method based on RUS circuits is used to approximate $R_Z(a)$ (instead of Selinger's method), then we obtain
\begin{equation}
\text{Exp}_Z^{\text{WK}}[T] = 2.52\log_2(10^{-\gamma}/\epsilon) + 1.14\log_2(10^\gamma) + 5.14
\enspace ,
\end{equation}
which is an improvement over the direct methods due to Selinger and KMM, even for large angles.
The density of the database presented in~\secref{sec:applications.database.axial} permits a maximum of $m = 6$ significant digits; a larger database would permit higher precision.

If full precision is required (i.e., $\gamma = 0$), then a slightly different method can be used.
Given an angle $\theta$ and error $10^{-6} > \epsilon \geq 10^{-11}$, an approximation of $R_Z(\theta)$ can be obtained by first using the RUS axial rotation database to get $R_Z(\tilde{\theta})$ such that $|\tilde{\theta} - \theta| = \phi \leq 10^{-6}$.  Then, a gearbox circuit can be used to approximate $\phi = a\times 10^{-\gamma}$ to within the prescribed distance $\epsilon$, where $R_Z(a)$ is obtained by again using the RUS database. The expected $T$ count is estimated by
\begin{equation}
\text{Exp}_Z^{\text{hybrid}}[T] = 1.26\log_2(1/\delta) + 2\cdot1.26\log_2(10^{-\gamma}/\epsilon) + 1.14\log_2(10^\gamma) + 1.61
\enspace ,
\end{equation}
where $\delta$ is the selected accuracy of the approximation $\tilde{\theta}$.  Assuming $\phi \approx \delta$ and therefore $10^\gamma \approx 10/\delta$, we obtain
\begin{equation}
\label{eq:estimated-hybrid-scaling}
\text{Exp}_Z^{\text{hybrid}}[T] \approx 2.52\log_2(1/\epsilon) - 0.12\log_2(1/\delta) - 2.97
\enspace .
\end{equation}
Thus, an effective strategy is to approximate $\theta$ to the maximum accuracy permitted by the axial RUS database ($\delta = 10^{-6}$) and then approximate the remaining angle $\phi$ with a gearbox circuit.

The coarse approximation $\tilde{\theta}$ will often be better than $10^{-6}$ so the actual scaling may vary from~\eqnref{eq:estimated-hybrid-scaling}.  To check, we calculated for $\epsilon \geq 10^{-11}$, the cost of the hybrid approach for the same $100$k angles used in~\secref{sec:applications.database.axial}.  The results yield an empirical fit of $2.62\log_2(1/\epsilon) - 3.1$, which is slightly higher than~\eqnref{eq:estimated-hybrid-scaling}, but still lower than that reported by KMM.

Even higher accuracy can be obtained by recursively applying the hybrid procedure.  If the mantissa $a$ of $\phi$ requires more accuracy than the RUS database can provide, then $R_Z(a)$ can be coarsely approximated using the database and the remainder can be obtained using another gearbox.  Asymptotically, such an approach has scaling $\Theta((1/\epsilon)^{1/\log_2(1/\delta)})$, making it practical only for a limited range of $\epsilon > 10\delta^2$.

\subsubsection{Results: Decomposition with non-axial rotations
\label{sec:applications.database.non-axial}
}

While it suffices to use three $Z$-axis rotations and two Hadamard gates to decompose an arbitrary single-qubit non-axial rotation, this process, used by \cite{Kliuchnikov2012b}, \cite{Selinger2012a} and \cite{Ross2014}, incurs a factor of three increase in cost, since each axial rotation must in turn be decomposed. This effect is illustrated in~\figref{fig:kmm-ratios} by the larger ratios for non-axial unitaries.
Using just our axial database for non-axial unitary decomposition results in a similar increase in cost.
Although Fowler's method~\cite{Fowl04c} does not incur the additional cost for arbitrary unitaries, maintaining a scaling of $2.95\log_2(1/\epsilon) + 3.75$, the method is exponential and does not achieve exact implementation for many unitaries.
RUS circuits, on the other hand, offer a large domain of exactly implementable unitaries.
As~\figref{fig:kmm-ratios} suggests, composing both axial and non-axial RUS circuits could yield better approximations than using $Z$-axis rotations alone.

Construction of the database in the non-axial case is significantly more challenging than in the axial case.  First, unitaries must be represented by three rotation angles instead of one.
Second, composition of circuits requires multiplication in the non-axial case, which is less efficient than for the $Z$-axis case which only requires addition.
Third, organization of the database to enable efficient lookup is more complicated;
$Z$-axis rotations can be sorted by rotation angle, while arbitrary unitaries require a more complicated data structure such as a $k$-d tree~\cite{Dawson2005,Amy2013a}.

However, we can express each unitary by its Clifford equivalence class representative~\eqnref{eq:clifford-equiv-representative}, and also avoid conjugating by all $576$ pairs of Clifford gates.
Since any single-qubit Clifford can be written as a product $g_1g_2$ where $g_1\in G_1$, $g_2 \in G_2$ and
\begin{equation}\begin{split}
G_1 &= \{I,Z,S,S^\dagger\}\\
G_2 &= \{I,H,X,XH,HS,XHS,HSH,XHSH\}
\enspace ,
\end{split}\end{equation}
then we need only conjugate by $G_2$.
Each resulting unitary can then be decomposed into three rotations
\begin{equation}
g_2 U g'_2 = R_Z(\theta_1)R_X(\theta_2)R_Z(\theta_3)
\enspace .
\end{equation}
The Clifford gates in $G_1$ are diagonal and only modify $\theta_1$ and $\theta_3$.  Up to conjugation by elements of $G_1$, we have
\begin{equation}
R_Z(\theta_1)R_X(\theta_2)R_Z(\theta_3) \equiv R_Z(\theta_1\!\!\!\mod \pi/2)R_X(\theta_2)R_Z(\theta_3\!\!\!\mod \pi/2)
\enspace .
\end{equation}
Choosing $0 \leq \theta_1, \theta_2 < \pi/2$, we can find an equivalence class representative without actually conjugating by $G_1$, saving a factor of $576/64 = 9$.

Using these optimizations, we construct a database of size $45526$ containing all RUS circuits with expected $T$ count at most $18$.
We calculated the best circuit for $100$ random single-qubit unitaries for a range of $\epsilon \geq 8\times 10^{-3}$.
A fit-curve of the data yields a scaling of $\text{Exp}_U[T] = 2.4\log_2(1/\epsilon) - 3.28$.
Based on the slope, the savings is roughly $18$ percent over Fowler;
in absolute terms, the savings is roughly a factor of two for modest approximation accuracy. See~\figref{fig:rus-nonaxial-decomposition}.
Given the relatively large ratios for non-axial unitaries in~\figref{fig:kmm-ratios} and the fact that our database contains only a limited subset of possible RUS circuits, by incorporating a larger set of circuits, we expect the scaling to further improve.

\begin{figure}
\centering
\includegraphics{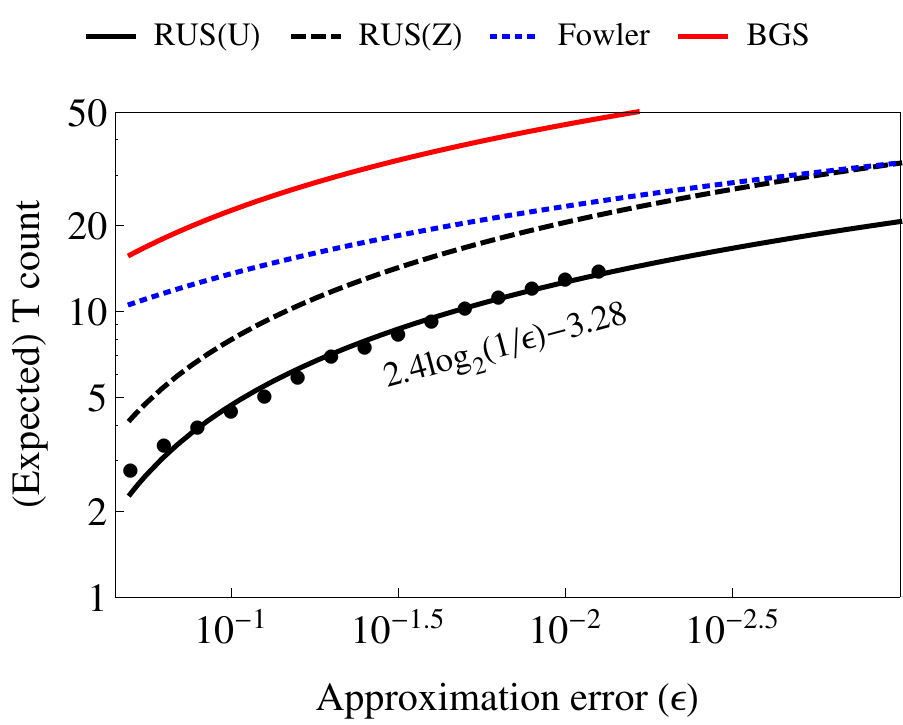}
\caption{\label{fig:rus-nonaxial-decomposition}
The expected number of $T$ gates required to approximate an arbitrary single-qubit unitary to within distance $\epsilon$.  Each point indicates the mean of $100$ random unitaries approximated to the corresponding accuracy with our full database of RUS circuits.  With $95$ percent confidence, the solid black line has slope in the range $[2.29,2.51]$. The dashed black line indicates the estimated cost of first expressing the unitary as a product of axial rotations, and then decomposing each rotation using the $Z$-axis RUS database from~\secref{sec:applications.database.axial}. The solid red line indicates the scaling obtained by using the circuit in~\figref{fig:v3-5.26} for $V_3$ decomposition~\cite{Bocharov2013}.  The scaling is worse than the others, but is valid for $\epsilon \geq 10^{-10}$. The estimated scaling using exponential direct search (Fowler~\cite{Fowl04c}) is shown for reference.
}
\end{figure}

\subsection{Quantum algorithms using coarse angles
\label{sec:applications.phase-estimation}
}
The accuracy of our decomposition method is limited by the size of the database.  Our $Z$-axis rotation database is capable of approximating arbitrary rotations up to an accuracy of $10^{-6}$.  To achieve higher accuracy, either the database must be expanded, or an algorithmic decomposition such as that of~\secref{sec:applications.v3} must be used.  However, a variety of important quantum algorithms require only limited rotation accuracies.  Fowler, for example, used numerical analysis to argue that Shor's algorithm requires rotation angles no smaller than $\theta = \pi/64\approx 0.05$ with an with an approximation error of $\epsilon = \pi/256 \approx 0.012$~\cite{Fowl03b}.

Another application of coarse angles is in quantum chemistry.  Consider a Hamiltonian for a molecule expressed in second quantized form, where the objective is to determine the ground state energy of the molecule.  Wecker et al.~\cite{Wecker2013} have developed a technique to obtain an estimate of the energy using only angles at most $10^{-6}$ accuracy in the phase estimation algorithm.
Similarly, Jones et al. show how to optimize quantum chemistry simulations by ignoring terms with small norm~\cite{Jones2012}. They use $Z$-axis rotations with approximation accuracies in the range $\epsilon = 10^{-5}$.
For such algorithms, our method produces rotations at the desired accuracy using extremely few $T$ gates.

\section{Conclusions and future work}
\label{sec:concl}
We have presented a general framework of non-deterministic circuits called ``Repeat-Until-Success" (RUS) circuits, and characterize unitaries which can be exactly represented by a RUS circuit.
Traditional methods decompose single-qubit unitaries into deterministic sequences of gates.
Wiebe and Kliuchnikov showed that by adding measurements and allowing non-deterministic circuits, decompositions with fewer $T$ gates are possible (in expectation) for very small $Z$-axis rotations \cite{Wiebe2013}.  Our results extend that conclusion to arbitrary single-qubit unitaries.  By synthesizing RUS circuits and then composing them, we can approximate arbitrary single-qubit unitaries to within a distance of $10^{-6}$, which is sufficient for many quantum algorithms. Approximation accuracy can be improved by combining our circuits with those of~\cite{Wiebe2013}.  For a random $Z$-axis rotation, our technique yields an approximation which requires as little as one-third as many $T$ gates as~\cite{Selinger2012a}, \cite{Kliuchnikov2012b}, \cite{Ross2014}, and~\cite{Fowl04c}.  Composing axial and non-axial RUS circuits yields even larger improvements in $T$ count costs, where the approximation accuracy is limited by the size of the database.

Our results suggest a number of possible areas for further research.
%First, the circuits proposed by~\cite{Wiebe2013} use traditional decomposition algorithms (i.e., Selinger or KMM) to generate the unitaries required for the mantissa $a$ of the angle $a \times 10^{-\gamma}$. Instead, RUS circuits could be used in order to improve performance.  Indeed, one could consider a hybrid approach that combines decomposition methods in order to find the most efficient circuit.
First, circuits of the form shown in~\figref{fig:two-cz-canonical} make up only a subset of possible RUS circuits.  Expanding the search to include additional types of circuits could improve database density.
Second, a formal number-theoretic characterization of RUS circuits needs to be made.
A theoretical understanding could lead to efficient decomposition algorithms based on RUS circuits and allow for approximation to much smaller values of $\epsilon$.

Extensions of the RUS circuit framework to multi-qubit unitaries or non-unitary channels should also be considered.  In addition, we have restricted the setting to recovery operations that are Clifford operators.  That restriction could be modified to allow for larger or alternative classes of operations.  On the other hand, fault-tolerance schemes based on stabilizer codes often permit the application of Pauli operators~\cite{Knill2004} at no cost.  Thus, it might be sensible to limit recovery operations to only tensor products of Paulis.

Finally, the non-deterministic nature of RUS circuits imposes some additional constraints on the overall architecture of the quantum computer. Many fault-tolerance schemes already use non-deterministic methods to implement certain gates.  But most of the non-determinism occurs ``offline'', without impacting the computational data qubits.  Since RUS circuits are ``online'', the time required to implement a given unitary cannot be determined in advance.  Such asynchronicity will require extensive placement and routing techniques and classical control logic.  Architecture-specific analysis will be required in order to concretely assess the benefits of using RUS circuits.

\section*{Acknowledgements}
The authors extend thanks to Vadym Kliuchnikov, Alex Bocharov, Nathan Wiebe, Yuri Gurevich, Andreas Blas, David Gosset and Cody Jones for helpful discussions, and to Dave Wecker for assistance with the implementation of the direct search.   Thanks also to Robin Kothari for suggesting the amplitude amplification technique.  AEP would like to thank Microsoft Research and the entire QuArC group for their hospitality.

\bibliographystyle{alpha-eprint}
\bibliography{library}

\end{document}